\documentclass[acmsmall,authorversion]{acmart}
\usepackage[utf8]{inputenc}
\usepackage[framemethod=tikz]{mdframed}
\usepackage{appendix}
\usepackage{dsfont}
\usepackage{enumerate}
\usepackage{bbm}
\usepackage{float}
\usepackage{enumitem}

\usepackage{algorithm}
\usepackage{verbatim}
\usepackage[noend]{algpseudocode}
\usepackage{xfrac}
\usepackage{thmtools}
\usepackage{thm-restate}
\usepackage{pifont}
\usepackage{soul}
\newcommand{\xmark}{\ding{53}}%
\newcommand{\cmark}{\ding{51}}%
\usepackage{array}
\usepackage{sidecap}
\usepackage{booktabs}
\newcolumntype{C}[1]{>{\centering\arraybackslash}p{#1}} 

\declaretheorem[name=Theorem]{theorem}
\declaretheorem[name=Lemma]{lemma}
\declaretheorem[name=Definition]{definition}

\declaretheorem[name=Example,style=definition]{example}

\newcommand{\norm}[1]{\left\lVert#1\right\rVert}
\definecolor{lightgray}{gray}{0.5}

\newcommand{\abs}[1]{\left|#1\right|}
\newcommand{\BigO}[1]{\ensuremath{\operatorname{\mathcal O}\left(#1\right)}}
\newcommand{\SmallO}[1]{\ensuremath{\operatorname{o}\left(#1\right)}}

\renewcommand{\vec}[1]{{\pmb{#1}}}
\newcommand{\card}[1]{\left|#1\right|}
\newcommand{\set}[1]{\left\{#1\right\}}

\newcommand{\spacedcdot}{\,\cdot\,}
\newcommand{\naturals}{\mathbb N}
\newcommand{\reals}{\mathbb{R}}
\newcommand{\x}{\vec{x}}
\newcommand{\dv}{\vec{\theta}}

\renewcommand{\u}{\vec{u}}
\newcommand{\X}{\mathcal{X}}
\newcommand{\I}{\mathcal{I}}
\renewcommand{\U}{\mathcal{U}}
\newcommand{\A}{\mathcal{A}}
\newcommand{\F}{\mathcal{F}}

\newcommand{\algoname}{\textsc{OHF}}
\newcommand{\R}{\mathcal{R}}
\newcommand{\umax}{u_{\mathrm{max}}}
\newcommand{\umin}{u_{\mathrm{min\vspace{-1em}}}}
\newcommand{\T}{\mathcal{T}}
\newcommand{\parentheses}[1]{\left(#1\right)}
\newcommand\vGroup[2]{\underset{#2}{\underbrace{#1} } }

\newcommand{\brackets}[1]{\left[#1\right]}
\newcommand{\diam}[1]{\,\mathrm{diam}\parentheses{#1}}
\renewcommand{\S}{\mathcal{S}}
\newcommand{\M}{\mathcal{M}}
\renewcommand{\C}{\mathcal{C}}

\newcommand{\E}{\mathcal{E}}
\renewcommand{\P}{\mathcal{P}}

\newcommand{\regret}{\mathfrak{R}_T \parentheses{F_{\alpha}, \vec\A}}
\newcommand{\VT}{\mathbb{V}_{\T}}
\newcommand{\WT}{\mathbb{W}_{\T}}
\newcommand{\barregret}{\bar{\mathfrak{R}}_T \parentheses{F_{\alpha}, \vec\A}}

\usepackage{array}
\newcolumntype{C}[1]{>{\raggedright\arraybackslash}p{#1}}

\newcommand{\SC}{\textsc{Cycle}}
\newcommand{\BT}{\textsc{Tree}}
\newcommand{\Grid}{\textsc{Grid}}
\newcommand{\Abilene}{\textsc{Abilene}}
\newcommand{\GEANT}{\textsc{GEANT}}
\newcommand{\lru}{\textsc{LRU}}
\newcommand{\lfu}{\textsc{LFU}}
\newcommand{\horizonalgoname}{\textsc{OHF}}
\newcommand{\slotalgoname}{\textsc{OSF}}

\newcommand{\OPT}{\textsc{OPT}}
\usepackage{titlesec}

\usepackage{tikz}
\usetikzlibrary{shapes.geometric,calc}

\newcommand\score[2]{

\pgfmathsetmacro\pgfxa{#1+1}

\tikzstyle{scorestars}=[circle, draw, inner sep=0.15em,anchor=west]

\begin{tikzpicture}[baseline]
  \foreach \i in {1,...,#2} {
    \pgfmathparse{(\i<=#1?"black!70":"white")}
    \edef\starcolor{\pgfmathresult}
    \draw (\i*1em,0) node[name=star\i,scorestars,fill=\starcolor]  {};
   }
  \pgfmathparse{(#1>int(#1)?int(#1+1):0}
  \let\partstar=\pgfmathresult
  \ifnum\partstar>0 
    \pgfmathsetmacro\starpart{#1-(int(#1))}
    \path [clip] (star\partstar.north west) rectangle 
    ($(star\partstar.south west)!\starpart!(star\partstar.south east)$);
    \fill (\partstar*1em,0) node[scorestars,fill=black!70]  {};
  \fi,

\end{tikzpicture}
}

\usepackage{soul}

\newcommand{\new}[1]{{#1}}

\usepackage{subcaption}

\title{Enabling Long-term  Fairness in Dynamic Resource Allocation}

\author{Tareq Si Salem}
\email{tareq.si-salem@inria.fr}
\affiliation{%
  \institution{Inria}
  \city{Sophia Antipolis}
  \country{France}
}
\author{George Iosifidis}
\email{G.Iosifidis@tudelft.nl}
\affiliation{%
  \institution{TU Delft}
  \city{Delft}
  \country{The Netherlands}
}

\author{Giovanni Neglia}
\email{giovanni.neglia@inria.fr}
\affiliation{%
  \institution{Inria}
  \city{Sophia Antipolis}
  \country{France}
}

\begin{abstract}
We study the fairness of dynamic resource allocation problem under the $\alpha$-fairness criterion. We recognize two different fairness objectives that naturally arise in this problem: the well-understood slot-fairness objective that aims to ensure fairness at every timeslot, and the less explored horizon-fairness objective that aims to ensure fairness across utilities accumulated over a time horizon.  We argue that horizon-fairness comes at a lower price in terms of social welfare.  We study horizon-fairness with the regret as a performance metric and show that vanishing regret cannot be achieved in presence of an unrestricted adversary. We propose restrictions on the adversary's capabilities corresponding to realistic scenarios and an online policy that indeed guarantees vanishing regret under these restrictions. We  demonstrate the applicability of the proposed fairness framework to a representative resource management problem considering a virtualized caching system where different caches cooperate to serve content requests.
\end{abstract}

\setcopyright{acmcopyright}
\acmJournal{POMACS}
\acmYear{2022} \acmVolume{6} \acmNumber{3} \acmArticle{46} \acmMonth{12} \acmPrice{15.00}\acmDOI{10.1145/3570606}

\received{August 2022}
\received[revised]{October 2022}
\received[accepted]{November 2022} 



\begin{document}

\maketitle
\section{Introduction}
Achieving fairness when allocating resources in communication and computing systems has been a subject of extensive research, 
and has been successfully applied in numerous practical problems. Fairness is leveraged to perform congestion control in the Internet~\cite{kellyORS98, mo2000fair}, to select transmission power  in multi-user wireless networks~\cite{srikantJSAC06, tassiulasFnT06}, and to allocate  multidimensional resources in cloud computing platforms~\cite{carleeToN13, baochuninfocom14, bonaldsigmetrics15}. Depending on the problem at hand, the criterion of fairness can be expressed in terms of how the service performance is distributed across the end-users, or in terms of how the costs are balanced across the servicing nodes. The latter case exemplifies the natural link between fairness and load balancing in  resource-constrained systems~\cite{loadbalacing1,loadbalacing2}. A prevalent fairness metric is $\alpha$-fairness, which encompasses the utilitarian principle (Bentham-Edgeworth solution~\cite{edgeworth1881mathematical}), proportional fairness (Nash bargaining solution~\cite{Nash1950}), max-min fairness (Kalai–Smorodinsky bargaining solution~\cite{kalai1975other}){, and, under some conditions, Walrasian equilibrium~\cite{iosifidis-sigmetrics15}.} All these fairness metrics have been used in different cases for the design of resource management mechanisms~\cite{radunovic2007unified, Nace2008}.


A common limitation of the above works is that they consider \emph{static} environments. That is, the resources to be allocated and, importantly, the users' utility functions, are fixed and known to the decision maker. This assumption is very often unrealistic for today's communication and computing systems. For instance, in small-cell mobile networks the user churn is typically very high and unpredictable, thus hindering the fair allocation of spectrum to cells~\cite{andrews5g}. Similarly, placing content files at edge caches to balance the latency gains across the served areas is non-trivial due to the non-stationary and fast-changing patterns of requests~\cite{paschosComMag16}. At the same time, the increasing virtualization of these systems introduces cost and performance volatility,  as extensive measurement studies have revealed~\cite{traverso2013temporal,leconte2016placing,elayoubi2015performance}. This uncertainty is exacerbated for services that process user-generated data (e.g., streaming data applications) where the performance (e.g., inference accuracy) depends also on  a priori unknown input data and  dynamically selected machine learning libraries~\cite{jose-conext,Liu2019Aug,Alipourfard2017Mar}. 

\subsection{Contributions}
This paper makes the next step towards enabling long-term fairness in dynamic systems. We consider a system that serves a set of agents $\mathcal I$, where a controller selects at each timeslot $t \in \naturals$ a resource allocation profile $\x_t$ from a set of eligible allocations~$\mathcal X$ based on  past agents' utility functions ${\u}_{t'}: \X \to \reals^\I$ for $t' < t$ and of $\alpha$-fairness function $F_{\alpha} : \reals^\I_{\geq 0} \to \reals$. The utilities might change due to unknown, unpredictable,  and (possibly) non-stationary perturbations that are revealed to the controller only after it decides $\x_t$. We employ the terms \emph{horizon-fairness}~(HF) and \emph{slot-fairness}~(SF) to distinguish the different ways fairness can be enforced  in a such time-slotted dynamic system. Under horizon-fairness, the controller enforces fairness on the aggregate utilities for a given time horizon~$T$, whereas under slot-fairness, it enforces fairness on the utilities at each timeslot separately. Both metrics have been studied in previous work, e.g., see~\cite{gupta2021individual,Liao2022Feb,jalota2022online,Sinclair2022} and the discussion in Section~\ref{s:related_work}. Our focus is on horizon-fairness, which raises novel technical challenges and subsumes slot-fairness as a special case.


We design the \emph{online horizon-fair}~(\algoname) policy by leveraging \emph{online convex optimization}~(OCO)~\cite{Hazanoco2016}, to handle this reduced-information setting under a powerful \emph{adversarial} perturbation model. \new{Adversarial analysis is a modeling technique to characterize a system's performance under unknown and hard to characterize exogenous parameters and has been recently successfully used to model caching problems (e.g., in~\cite{paschos2019learning, sisalem2021no, mhaisen2022online, paria2021texttt,bura2021learning,Li2021,salem2021accai}).
In our context}, the performance of a resource allocation policy $\vec \A$ is evaluated by the \emph{fairness regret}, which is defined as the difference between the $\alpha$-fairness, over the time-averaged utilities, achieved by a static optimum-in-hindsight (\emph{benchmark}) and the one achieved by the policy:
\begin{align}
    \regret \triangleq \sup_{ \set{\u_t}^T_{t=1} \in {{\U^T}}} \set{\max_{\x \in \X}  F_{\alpha}\left({\frac{1}{T}\sum_{t \in \T}\vec u_{t}(\x)}\right) -F_{\alpha}\left({\frac{1}{T}\sum_{t \in \T}\vec u_{t}(\x_t)}\right)}.
    \label{e:b_regret0}
\end{align}
If the fairness regret vanishes over time (i.e., $\lim_{T\to \infty}  \regret = 0$),  policy $\vec \A$ will attain the same fairness value as the static benchmark under any possible sequence of utility functions. A policy that achieves sublinear regret under these  adversarial conditions, can also succeed in more benign conditions where the perturbations are not adversarial, or the utility functions are revealed at the beginning of each slot.

The fairness regret metric~\eqref{e:b_regret0}  departs from the  template of~OCO. In particular, the scalarization of the vector-valued utilities, through the $\alpha$-fairness function, is not applied at every timeslot to allow for the controller to easily adapt its allocations, instead is only applied at the end of the time horizon $T$.  Our first result characterizes the challenges in tackling this learning problem. Namely, Theorem~\ref{theorem:impossibility} proves that, when utility perturbations are only subject to four mild technical conditions, such as in standard OCO, it is impossible to achieve vanishing fairness-regret. Similar negative results were obtained under  different setups of  primal-dual learning and online saddle point learning~\cite{mannor2009online,anderson2022lazy, rivera2018online}, but they have been devised for specific problem structures (e.g., online matrix games)  and thus do not apply to our setting.


In light of this negative result, we introduce additional \emph{necessary} conditions on the adversary to obtain a vanishing regret guarantee. Namely, the adversary can only induce perturbations to the time-averaged utilities we call budgeted-severity or partitioned-severity constrained. These conditions capture several practical utility patterns, such as non-stationary corruptions, ergodic and periodic inputs~\cite{Liao2022Feb,balseiro2022best,zhou2019robust, duchi2012ergodic}. We proceed to propose the \algoname{} policy which adapts dynamically the allocation decisions and provably achieves $\regret = o(1)$ (see Theorem~\ref{th:maintheorem}). 




The \algoname{} policy employs a novel learning approach that operates concurrently, and in a synchronized fashion, in a primal and a dual (conjugate) space.  Intuitively, \algoname{} learns the weighted time-varying utilities in a primal space, and learns the weights accounting for the global fairness metric in some dual space. To achieve this, we develop novel  techniques through a convex conjugate approach (see Lemmas~\ref{lemma:convex_conjugate},~\ref{l:recover_f}, and~\ref{l:saddle_problem} in the Appendix). 

Finally, we  apply our fairness framework to a  representative resource management
problem in virtualized caching systems where different caches cooperate by serving jointly the
received content requests. We evaluate  the performance of \algoname{} with its slot-fairness counterpart policy through numerical examples. We evaluate the price of fairness of \algoname,  which quantifies the efficiency loss due to fairness, across different network topologies and participating agents.  Lastly, we apply \algoname{} to a Nash bargaining scenario, a concept that has been widely used in resource allocation to distribute to a set of agents the utility of their cooperation~\cite{Boche2011,Iosifidis2017,Wenjie2009,Liang2017}.

\subsection{Outline of Paper}
The paper is organized as follows. The related literature is discussed in Section~\ref{s:related_work}. The definitions and background are provided in Section~\ref{s:fairness}. The adversarial model and the proposed algorithm are presented in Section~\ref{s:OHF}. Extensions to the fairness framework are provided in Section~\ref{s:extensions}.  The resource management
problem in virtualized caching systems application is provided in Section~\ref{s:experiments}. Finally, we conclude the paper and provide directions for future work in Section~\ref{s:conclusion}.

\section{Literature Review}
\label{s:related_work}

\subsection{Fairness in Resource Allocation}
Fairness has found many applications in wired and wireless networking~\cite{kellyORS98, mo2000fair,srikantJSAC06,tassiulasFnT06,altman2008generalized}, and cloud  computing platforms~\cite{carleeToN13, baochuninfocom14, bonaldsigmetrics15}.  Prevalent fairness criteria are the max-min fairness and proportional fairness, which are rooted in axiomatic bargaining theory, namely the Kalai–Smorodinsky~\cite{kalai1975other} and Nash bargaining solution~\cite{Nash1950}, respectively. On the other side of the spectrum, a controller might opt to ignore fairness and maximize the aggregate utility of users, i.e., to follow the \emph{utilitarian principle}, also referred to as the  Bentham-Edgeworth solution~\cite{edgeworth1881mathematical}. The \emph{Price of Fairness} (PoF)~\cite{bertsimas2011price} is now an established metric for assessing how much the social welfare (i.e., the aggregate utility) is affected when enforcing some fairness metric. Ideally, we would like this price to be as small as possible, bridging in a way these two criteria. Atkinson~\cite{ATKINSON1970244} proposed  the unifying  $\alpha$-fairness criterion which yields different fairness criteria based on the  value of~$\alpha \in \reals_{\geq 0}$, i.e.,  the utilitarian principle~($\alpha=0$), proportional fairness~($\alpha =1$), and max-min fairness~($\alpha\to \infty$). Due to the generality of the $\alpha$-fairness criterion, we use  it to develop our theory,  which in turn renders our results transferrable to all above fairness and bargaining problems. In this work, the PoF, together with the metric of fairness-regret, are the two criteria we use to characterize our fairness solution. 

\subsection{Fairness in Dynamic Resource Allocation} 
Several works consider slot-fairness in dynamic systems~\cite{jalota2022online,Sinclair2022, Talebi2018}. Jalota and Ye~\cite{jalota2022online} proposed a weighted proportional fairness algorithm for a system where new users arrive in each slot, having linear i.i.d. perturbed unknown utility functions at the time of selecting an allocation, and are allocated resources from an i.i.d. varying budget.  Sinclair et al.~\cite{Sinclair2022} consider a similar setup, but assume the utilities are known at the time of selecting an allocation, and the utility parameters (number of agents and their type) are drawn from some fixed known distribution. They propose an adaptive threshold policy, which achieves a target efficiency (amount of consumed resources' budget) and fairness tradeoff, where the latter is defined w.r.t. to an offline weighted proportional fairness benchmark.
Finally, Talebi and  Proutiere~\cite{Talebi2018} study dynamically arriving tasks that are assigned to a set of servers with unknown and stochastically-varying service rates. Using a stochastic multi-armed bandit model, the authors achieve proportional fairness across the service rates assigned to different tasks at each slot. All these important works, however, do not consider the more practical horizon-fairness metric where fairness is enforced throughout the entire operation of the system and not over each slot separately. 

Horizon-fairness has been recently studied through the lens of competitive analysis~\cite{kawase2021online, banerjee2022online, bateni2022fair}, where the goal is to design a policy that achieves online fairness within a constant factor from the fairness of a suitable benchmark. Kawase and Sumita~\cite{kawase2021online} consider the problem of  allocating arriving items irrevocably to one agent who has additive utilities over the items. The arrival of the items is arbitrary and can even be selected by an adversary. The authors consider known utility at the time of  allocation, and design policies under the max-min fairness criterion. Banerjee et al.~\cite{banerjee2022online} consider a similar problem  under the proportional fairness criterion, and they allow the policies to exploit available predictions. We observe that the competitive ratio guarantees, while theoretically interesting, may not be informative about the fairness of the actual approximate solution achieved by the algorithm for ratios different from one. For instance, when maximizing a Nash welfare function under the proportional fairness criterion, the solution achieves some axiomatic fairness properties~\cite{Nash1950} (e.g., Pareto efficiency, individual rationality, etc.), but this welfare function is  meaningless for ``non-optimal'' allocations~\cite{Sinclair2022}, i.e., a policy with a high competitive ratio  is not necessary less fair than a policy with a lower competitive ratio. For this reason, our work considers regret as a performance metric: when regret vanishes asymptotically, the allocations of the policy indeed achieve the exact same objective as the adopted benchmark.

\new{Altman et al.~\cite{altman2012multiscale} consider the $\alpha$-fairness problem in a dynamic resource allocation, and investigate  fairness enforced at different time scales (instantaneous and long-term). They consider known utilities at the time of selecting an allocation in a stationary setting. Lodi et al.~\cite{lodi2021fairness} also treat  fairness across different time scales (single-period and $T$-period) as an offline problem.  In this work, we make a similar distinction on the fairness criterion in the online setting, where we define the  slot-fairness and horizon-fairness.}

A different line of work~\cite{gupta2021individual,Liao2022Feb,Cayci2020,benade2018make,zeng2020fairness,sinclair2020sequential,baek2021fair} considers horizon-fairness through regret analysis. Gupta and Kamble~\cite{gupta2021individual} study individual fairness criteria that advocate similar individuals should be treated similarly. They extend the notion of individual fairness to online contextual decision-making, and  introduce: (1)~fairness-across-time and (2)~fairness-in-hindsight. Fairness-across-time criterion requires the treatment of individuals to be individually fair relative to the past as well as future, while  fairness-in-hindsight only requires individual fairness at the time of the decision. The utilities are known at the time of selecting an allocation and are i.i.d. and drawn from an unknown fixed distribution. Liao et al.~\cite{Liao2022Feb} consider a similar setup to ours, with a limited adversarial model and time-varying but {known} utilities, and focus on  proportional fairness.
They consider  adversarial perturbation added on a fixed item distribution where the  demand of items generally behaves predictably, but for some time steps, the demand behaves erratically. 
Our approach departs significantly from these interesting works in that we consider unknown utility functions, a broader adversarial model (in fact, as broad as possible while still achieving vanishing fairness regret), and by using the general $\alpha$-fairness criterion that encompasses all the above criteria as special cases. This makes, we believe, our \algoname{} algorithm applicable to a wider range of practical problems.  Table~\ref{tab:related_work} summarizes the differences between our contribution and the related works.

\begin{table}[t]
	\caption{Summary of related work under online fairness in resource allocation.}
\begin{footnotesize}
	\begin{center}
		\begin{tabular}{clc>{\centering\arraybackslash}p{4.6em} >{\centering\arraybackslash}p{4.6em} c}
			\hline
			\textbf{Paper} & \textbf{Criterion}&\textbf{HF/SF} & \textbf{Unknown utilities} & \textbf{Adversarial utilities} & \textbf{Metric}\\
			\hline
			\cite{jalota2022online} & Weighted proportional fairness & SF & \cmark& \xmark& Regret  \\
			\cite{Sinclair2022} & Weighted proportional fairness & SF &\xmark&\xmark& Envy, Efficiency\\
			\cite{Talebi2018} & Proportional fairness& SF &\xmark&\xmark &  Regret \\
			\cite{gupta2021individual} & Individual fairness& HF/SF &\xmark&\xmark& Regret  \\
			\cite{Liao2022Feb} & Proportional fairness& HF&\xmark& \cmark&  Regret \\
			\cite{Cayci2020} & $\alpha$-fairness &HF & \cmark& \xmark& Regret  \\
			\cite{benade2018make} & Envy-freeness & HF&\xmark&\cmark& Envy  \\
			\cite{zeng2020fairness}  & Weighted proportional fairness & HF&\xmark&\cmark& Envy, Pareto Efficiency  \\
			\cite{baek2021fair} & Proportional fairness & HF&\xmark&\xmark& Regret\\
			\cite{kawase2021online} & Max-Min fairness  & HF& \xmark &\cmark& Competitive ratio  \\
			\cite{banerjee2022online} & Proportional fairness & HF&\xmark& \cmark& Competitive ratio   \\
			\cite{bateni2022fair} & Proportional fairness & HF &\xmark&\xmark& Competitive ratio   \\
			\hline
		This work & Weighted $\alpha$-fairness & HF/SF & \cmark& \cmark& Fairness Regret\\
			\hline
		\end{tabular}
	\end{center}
	\end{footnotesize}
	\label{tab:related_work}
\end{table} 


\subsection{Online Learning}
Achieving horizon-fairness  in our setup requires technical extensions to the theory of OCO~\cite{Hazanoco2016}. The basic template of OCO-learning (in terms of resource allocation) considers that a decision maker selects repeatedly a vector $\x_t$ from a convex set $\X$, before having access to the $t$-th slot scalar utility function $u_t(\x)$, with the goal to maximize the aggregate utility $\sum^T_{t=1} u_t(\x_t)$. The  decision maker aims to have vanishing time-averaged regret, i.e., the time-averaged distance of the aggregate utility $\sum^T_{t=1} u_t(\x_t)$ from the aggregate utility of the optimal-in-hindsight allocation $\max_{\x \in \X}\sum_{t=1}^T u_t(\x)$ for some time horizon $T$. OCO models are robust, expressive, and can be tackled with several well-studied learning algorithms~\cite{Hazanoco2016,ShalevOnlineLearning, mcmahan2017survey}. 
However, none of those is suitable for the fairness problem at hand, as we need to optimize a global function $F_\alpha(\spacedcdot)$ of the time-averaged vector-valued utilities. 
This subtle change creates additional technical complications. Indeed, optimizing functions of time-averaged utility/cost functions in learning is an open and challenging problem. In particular, Even-Dar et al.~\cite{evan2009colt} introduce the concept of global functions in online learning, and devise a policy with vanishing regret using the theory of approachability~\cite{blackwell1956analog}. However, their approach can handle only norms as global functions, and this limitation is not easy to overcome: the authors themselves stress that characterizing when  a global function enables a vanishing regret is an open problem (see~\cite[Section~7]{evan2009colt}). Rakhlin et al.~\cite{rakhlin11} extend this work to non-additive global functions. However, the $\alpha$-fairness function considered in our work is not supported by their framework.  To generalize the results to $\alpha$-fairness global functions, we employ a convex conjugate approach  conceptually similar  to the approach taken in the work of Agrawal and Evanur~\cite{agrawal2014bandits} to obtain a regret guarantee  with a concave global function under a stationary setting and linear utilities. In this work, we consider an adversarial setting (i.e., utilities are picked by an adversary after we select an allocation) that encompasses general concave utilities, and this requires learning over the primal space as well as the dual (conjugate) space.

\section{Online Fairness: Definitions and Background}
\label{s:fairness}
\subsection{Static Fairness}
\label{s:fairness+static}
Consider a system $\S$ that serves a set of agents $\I$ by selecting allocations from the set of eligible allocations $\X$.\footnote{\new{ Appendix~\ref{appdendix:departingarriving} discusses the setting in which the set of agents $\I$ is unknown and agents can depart and arrive to the system.}} 
In the general case, this set is defined as the Cartesian product of agent-specific eligible allocations' set $\X_i$, i.e., $\X \triangleq \bigtimes_{i \in \I} \X_i$. We assume that each set $\X_i$ is convex. The utility  of each agent $i \in \I$ is a concave function $ u_i: \X \to \reals_{\geq 0}$, and depends, possibly, not only on $\x_i \in \X_i$, but on the entire vector $\x\in \X$.\footnote{For example, in TCP congestion control, the performance of each end-node depends not only on the rate that is directly allocated to that node, but also, through the induced congestion in shared links, by the rate allocated to other nodes~\cite{kellyORS98}. Similar couplings arise in wireless transmissions over shared channels~\cite{tassiulasFnT06}.}
The vector $\vec u(\x) \triangleq \parentheses{u_{i} (\x)}_{i \in \I} \in \U$ is the vectorized form of the agents' utilities, where  $\U$ is the set of possible utility functions.  The joint allocation $\x_{\star} \in \X$ is an $\alpha$-fair allocation for some $\alpha \in \reals_{\geq 0}$ if it solves the following convex problem:
\begin{align}
    \max_{\x \in \X} F_{\alpha}\parentheses{\vec u (\x)},\label{e:prob2}
\end{align}
where $F_{\alpha}$ is the $\alpha$-fairness criterion the system employs (e.g., when $\alpha = 1$, problem~\eqref{e:prob2} corresponds to an Eisenberg-Gale convex problem~\cite{eisenberg1959consensus}). The $\alpha$-fairness function is defined as follows~\cite{ATKINSON1970244}:
\begin{definition}
An $\alpha$-fairness  function $F_\alpha:\U\to \reals$ is parameterized by the inequality aversion parameter $\alpha \in \reals_{\geq 0}$, and it is given by \begin{align}
     F_\alpha(\u) &\triangleq \sum_{i \in \I} f_\alpha(u_i), &\text{where}\qquad f_\alpha (u) \triangleq \begin{cases}
    \frac{u^{1-\alpha} - 1}{1-\alpha},&  \text{ for $\alpha \in \reals_{\geq 0} \setminus \set{1}$},\\
    \log(u),&  \text{ for $\alpha = 1$},
    \end{cases}\label{e:alpha-fair}
\end{align}
for every $\u \in \U$. Note that $\U \subset \reals^\I_{\geq 0}$ for $\alpha <1$, and $\U \subset \reals^\I_{>0}$ for $\alpha \geq 1$.
\end{definition}
Note that we use the most general version of utility-based fairness where the fairness rule is defined w.r.t. to accrued utilities (as opposed to allocated resource, only), i.e., in our system $\S$, the utility vector $\u \in \U$ can be a function of the selected allocations in~$\X$.   The $\alpha$-fairness function is concave and component-wise increasing, and thus exhibits diminishing returns~\cite{bertsimas2012efficiency}. An increase in utility to a player with a low utility results in a higher $\alpha$-fairness objective. Thus, such an increase is desirable to the system controller. Moreover,  the rate at which 
the marginal increase diminishes 
is controlled by  $\alpha$, which is then called the
\emph{inequality aversion parameter}. 
\new{An allocation which maximizes the $\alpha$-fairness objective is always Pareto efficient~\cite{bertsimas2012efficiency}.}
\subsection{Online Fairness}
We consider 
the performance of the system~$\S$ is tracked over a time horizon spanning $T\in\naturals$ timeslots. At the beginning of each timeslot $t \in \T \triangleq \set{1,2, \dots, T}$, a policy selects an allocation $\x_t \in \X$ \emph{before} $\u_{t}: \X \to \reals^\I$ is revealed to the policy. The goal is to approach the performance of a properly-selected fair allocation benchmark. We consider the following two cases:


\paragraph{Slot-Fairness.} An offline benchmark in hindsight, with access to the utilities revealed at every timeslot $t \in \T$, can ensure fairness at every timeslot satisfying a \emph{slot-fairness} (SF) objective~\cite{jalota2022online,Sinclair2022, Talebi2018}. Formally, the benchmark selects the joint allocation $\x_{\star} \in \X$ satisfying 
\begin{align}
 \mathrm{SF}:\qquad   \x_\star \in \arg\max_{\x \in \X}\frac{1}{T} \sum_{t \in \T} F_{\alpha}\parentheses{\vec u_{t}(\x)}.\label{eq:sf_objective}
\end{align}
\paragraph{Horizon-Fairness.}  Enforcing fairness at every timeslot can be quite restrictive, and this is especially evident for large time horizons.  An alternative formulation is to consider that the agents can accept  a momentary violation of fairness  at a given timeslot $t \in \T$ as long as in the long run fairness over the total incurred utilities is achieved. Therefore, it is more natural (see Example~\ref{e:example}) to ensure a horizon-fairness criterion over the entire period $\T$.
Formally, the benchmark selects the allocation $\x_{\star} \in \X$ satisfying
\begin{align}
   \mathrm{HF}:\qquad   \x_{\star} \in \arg\max_{\x \in \X} F_{\alpha}\parentheses{\frac{1}{T}\sum_{t \in \T}\vec u_{t}(\x)}.\label{eq:hf_objective}
\end{align}
\paragraph{Price of fairness.} Bertsimas et al.~\cite{bertsimas2012efficiency} defined the \emph{price of fairness} (PoF) metric to quantify the efficiency loss due to fairness as  the   difference   between   the maximum system efficiency and the efficiency  under  the  fair  scheme. In the case of $\alpha$-fairness, it is defined for some utility set $\U$ as 
\begin{align}
   \mathrm{PoF} (\U; \alpha) \triangleq \frac{ \max_{\u \in \U}  F_{0} (\u) -   F_{0} \parentheses{\u_{\max, \alpha}}}{ \max_{\u \in \U}  F_{0} (\u) }, 
\end{align}
where $\u_{\max, \alpha} \in  \arg\max_{\u \in \U } F_{\alpha} (\u)$ and $ F_0(\u)=\sum_{i \in \I} u_i$ measures  the achieved social welfare. Note that by definition the utilitarian objective achieves maximum efficiency, i.e.,  $\mathrm{PoF} (\U; 0)=0$.  Naturally, in our online setting, the metric is extended as follows
\begin{align}
       \mathrm{PoF} (\X; \T; \alpha) \triangleq \frac{ \max_{\x\in \X} \sum_{t \in \T} F_{0} \parentheses{\u_t (\x)}-   \sum_{t \in \T} F_{0} \parentheses{\u_t (\x_\star)}}{  \max_{\x\in \X} \sum_{t \in \T} F_{0} \parentheses{\u_t (\x)}},
\end{align}
where $\x_\star$ is obtained through either SF~\eqref{eq:sf_objective} or HF~\eqref{eq:hf_objective}. We provide the following example to further motivate our choice of horizon-fairness as a performance objective. A similar argument is provided in~\cite[Example 7]{lodi2021fairness}.

\begin{example}
\label{e:example}
 Consider a system with two agents $\I = \set{1,2}$, an allocation set $\X = [0, x_{\max}]$ with $x_{\max} > 1$, $\alpha$-fairness criterion with $\alpha =1$, even $T \in \naturals$, and the following sequence of utilities $\{\u_{t} (x)\}^T_{t=1} = \set{\parentheses{1+x, 1-x}, \parentheses{1+x, 1+x},\dots}$. It can easily be verified that $\mathrm{PoF} = 0$ for HF objective~\eqref{eq:hf_objective} because the HF optimal allocation is $x_{\max}$ which matches the optimal allocation under the utilitarian objective. However, under the SF objective~\eqref{eq:sf_objective} we have $\mathrm{PoF} = \frac{x_{\max} - 0.5}{ x_{\max} + 2} \approx 1$ when $x_{\max}$ is large. Remark that the two objectives have different domains of definitions; in particular, the allocations in the set $[1, x_{\max}] \subset \X$ are unachievable by the SF objective because they would lead to $u_{t,2}(x) \leq 0$. The HF  objective achieves lower PoF (hence, larger aggregate utility), and it allows a much larger set of eligible allocations (in particular all the allocations in the set $\X$), as shown in Fig.~\ref{e:example_motivating}. Indeed, when the controller has the freedom to achieve fairness over a time horizon, there is an  opportunity for more efficient allocations during the system operation. This example provides intuition on the robustness and practical importance of the horizon-fairness objective. 
\begin{figure}[t]
    \centering
    \includegraphics[width =240pt]{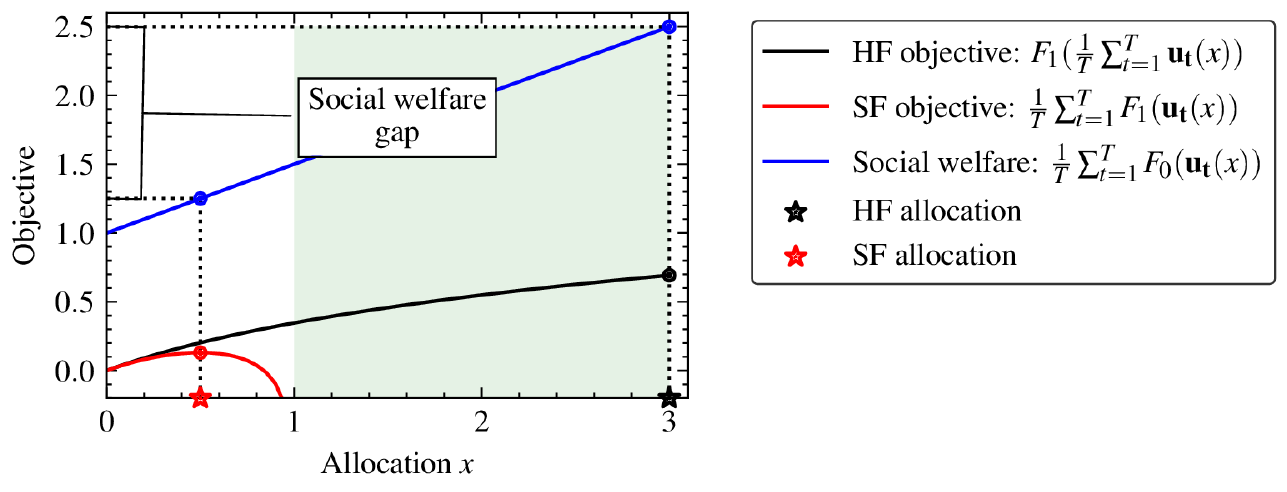}
        \caption{\protect\rule{0ex}{-10ex}Price of Fairness under HF and SF objectives for Example~\ref{e:example} for $x_{\max} =3$. The green shaded area provides the set of allocation unachievable by the SF objective but achievable by the HF objective.}
    \label{e:example_motivating}
\end{figure}
\end{example}

In the following section, we provide the description of an online learning model and our performance metric of interest under the HF objective.  
\subsection{Online Policies and Performance Metric}
\label{s:system_model}
The agents' allocations are determined by an online policy $\vec\A = \set{\A_1, \A_2, \dots, \A_{T}}$, i.e., a sequence of mappings. For every timeslot $t \in \T$, $\A_t: \X^t \times \U^t \to \X$ maps the sequence of past allocations~$\set{\x_s}^{t}_{s=1} \in \X^{t}$ and utility functions $\set{\u_s}^t_{s=1} \in \U^t$  to the next allocation $\x_{t+1} \in \X$. We assume the initial decision $\x_1$ is feasible (i.e., $\x_1 \in \X$). We measure the performance of policy $\vec\A$ in terms of the \emph{fairness regret}~\eqref{e:b_regret}, i.e., the difference between the fairness objective experienced by $\vec\A$ at the time horizon $T$ and that of the best static decision $\x_{\star} \in \X$ in hindsight.
We restate the regret metric here to streamline the presentation:
\begin{align}
    \regret \triangleq \sup_{ \set{\u_t}^T_{t=1} \in {{\U^T}}} \set{ F_{\alpha}\left({\frac{1}{T}\sum_{t \in \T}\vec u_{t}(\x_{\star})}\right) -F_{\alpha}\left({\frac{1}{T}\sum_{t \in \T}\vec u_{t}(\x_t)}\right)}.
    \label{e:b_regret}
\end{align}
where $\x_{\star}$ is the HF~\eqref{eq:hf_objective} allocation. 
If the fairness regret becomes negligible for large $T$, then  $\vec\A$ attains the same fairness objective as the optimal static decision with hindsight.
Note that under the utilitarian objective ($\alpha=0$), this fairness regret  coincides with the classic time-averaged regret in OCO~\cite{Hazanoco2016}. However, for general values of $\alpha \neq 0$, the metric is completely different, as we aim to compare $\alpha$-fair functions evaluated at time-averaged vector-valued utilities.
\section{Online Horizon-Fair (\algoname) Policy}
\label{s:OHF}
We first present in Section~\ref{s:adversarial_model}, the adversarial model considered in this work and provide a result on the  impossibility of guaranteeing vanishing fairness regret~\eqref{e:b_regret} under general  adversarial perturbations. We also provide a powerful family of adversarial perturbations for which a vanishing fairness regret guarantee is attainable. Secondly, we present the \algoname{} policy in Section~\ref{s:algorithm_description}  and provide its performance guarantee. Finally, we provide in Section~\ref{s:examples} a set of adversarial examples captured by our fairness framework.

\subsection{Adversarial Model and Impossibility Result}
\label{s:adversarial_model}
We begin by introducing formally the adversarial model that characterizes the utility perturbations. In particular, we consider $\vec\delta_t(\x) \triangleq { \left( \frac{1}{T}\sum_{s \in \T}\u_s(\x) \right) - \u_t(\x)}$ to quantify how much the adversary \emph{perturbs} the average utility  by selecting a utility function $\u_t$ at timeslot $t \in \T$. Recall that  $\x_{\star} \in \X$ denotes  the optimal allocation under HF objective~\eqref{eq:hf_objective}. We denote by $\Xi(\T)$ the set of all possible decompositions of $\T$ into sets of contiguous timeslots, i.e., for every $\set{\T_1, \T_2, \dots, \T_K} \in \Xi(\T)$ it holds $\T= \dot \bigcup_{k \in \set{1,2,\dots, K}} \T_k$ and  $\max \T_k < \min \T_{k+1}$ for $k \in \set{1,2, \dots, K-1}$. We define two types of adversarial perturbations:
 \begin{align}
     &\!\!\!\!\!\text{Budgeted-severity: }&& \!\!\!\!\VT\!\triangleq\!\!\!\!\!\!  \sup_{ \set{\u_t}^T_{t=1} \in {{\U^T}}} \set{\sum_{t \in \T} \sum_{i \in \I} \abs{\delta_{t,i} (\x_{\star})} },\label{e:adv1} \\
       &\!\!\!\!\!\text{Partitioned-severity: }&& \!\!\!\!\WT \! \triangleq\!\!\!\!\!\! \sup_{ \set{\u_t}^T_{t=1} \in {{\U^T}}} \! \set{ \inf_{\substack{\set{\T_1,\T_2,\dots,\T_K}\\ \in\, \Xi(\T)}}\! \set{ \sum^K_{k=1} \sum_{i \in \I} \abs{\sum_{t \in \T_k}\! \!  \delta_{t,i}(\x_{\star})}\! +\! \sum^K_{k=1}\! \frac{\card{\T_k}^2}{\sum_{k' < k } \card{\T_k}+1}}}\label{e:adv2}.
 \end{align}
Our result in Theorem~\ref{th:maintheorem} implies that when either  $\VT$ or $\WT$ grows sublinearly in the time horizon (i.e., the perturbations satisfy at least one of these two conditions),  the regret of \algoname{} policy in Algorithm~\ref{alg:primal_dual_ogaogd} vanishes over time.  \new{We provide a detailed description of conditions~\eqref{e:adv1} and \eqref{e:adv2} below.}

The \emph{budgeted-severity} $\VT$ in Eq.~\eqref{e:adv1} bounds the total amount of perturbations of the time-averaged utility. When  $\VT = 0$ the adversary is only able to select a fixed function, otherwise the adversary is able to select time-varying utilities, while  keeping the total deviation no more than $\VT$. Moreover, the adversary is allowed to pick opportunely the timeslots to maximize performance degradation for the controller. This model is similar to the \emph{adversarial corruption} setting considered in~\cite{Liao2022Feb,balseiro2022best}, and it captures realistic scenarios where the utilities incurred at different timeslots are predictable, but can be perturbed for some fraction of the timeslots.
For instance, Internet traffic may experience spikes due to breaking news or other unpredictable events~\cite{esfandiari2015online}. 

The \emph{partitioned-severity} $\WT$ \new{in Eq.~\eqref{e:adv2}} may at first be less easy to understand than budgeted-severity condition~\eqref{e:adv1}, but is equally important from a practical point of view. For simplicity, consider a uniform decomposition of the timeslots, i.e., $\T_k = M$ for every $k \in \set{1,2,\dots, T/M}$ assuming w.l.g. $M$ divides $T$. Then the r.h.s.  term in Eq.~\eqref{e:adv2} can be bounded as follows:
\begin{align}
\sum^{T/M}_{k=1} \frac{\card{\T_k}^2}{\sum_{k' < k } \card{\T_k}+1}  = \sum^{T/M}_{k=1} \frac{M^2}{M (k-1)+1} =\BigO{M^2 + M \log(T / M)}.\label{e:uniform-partitioned-severity}
\end{align}
Hence, when $M = o(\sqrt{T})$  it holds $\sum^{T/M}_{k=1} \frac{\card{\T_k}^2}{\sum_{k' < k } \card{\T_k}+1} = o(T)$. Since this term grows sublinearly in time, it remains to characterize the growth of the l.h.s. term $ \sum^K_{k=1} \sum_{i \in \I} \abs{\sum_{t \in \T_k}  \delta_{t,i} (\x_{\star})}$  in Eq.~\eqref{e:adv2}. This term is related to the perturbations selected by the adversary, however the absolute value is only evaluated at the end of each contiguous subperiod $\T_k$, i.e., the positive and negative deviations from the average utilities can cancel out. For example, a periodic selection of utilities from some set with cardinality $M$ would have zero deviation for this term. This type of adversary is similar to the periodic adversary considered in~\cite{duchi2012ergodic,balseiro2022best}, but also includes adversarial selection of utilities from some finite set (see  Example~\ref{example:periodic} in Section~\ref{s:examples}). The partitioned-severity adversary can model  real-life applications that exhibit seasonal properties, e.g., the traffic may be completely different throughout the day, but daily traffic is self-similar~\cite{zhou2019robust}. This condition also unlocks the possibility to obtain high probability guarantees under stochastic utilities (see Corollary~\ref{corollary:stochastic}).


We formally make the following assumptions:
\begin{enumerate}[label=(A\arabic*)]
    \item  \label{a:1}  The allocation set $\X$ is convex with diameter  $\diam{\X} < \infty$.
    \item The utilities are bounded, i.e., $\u_{t}(\x)  \in \brackets{\umin, \umax}^{\I} \subset \reals^{\I}$ for every $t \in \T$. 
    \item The supergradients of the utilities are bounded over $\X$, i.e., it holds $\norm{\vec g}_2 \leq L_{\X} < \infty$ for any $\vec g \in \partial_{\x} u_{t,i}(\x)$ and  $\x \in \X$.  
    \item\label{a:4} The average utility of the optimal allocation~\eqref{eq:hf_objective} is bounded such that $\frac{1}{T} \sum_{t \in \T}\u_t(\x_{\star}) \in \brackets{u_{\star, \min}, u_{\star, \max}}^{\I} \subset \reals^{\I}_{>0}$.
    \item\label{a:5} The adversary is restricted to select utilities such that 
\begin{align}
    \min\set{\VT, \WT} = o(T).
\end{align}
\end{enumerate}
We first show that an adversary solely satisfying the mild assumptions~\ref{a:1}--\ref{a:4} can arbitrarily degrade the performance of any policy $\vec\A$. Formally, we have the following negative result:
\begin{theorem}
When Assumptions~\ref{a:1}--\ref{a:4} are satisfied, there is no online policy $\vec\A$ attaining $\regret  = \SmallO{1}$  for $\card{\I} > 1$ and $\alpha > 0$. Moreover, there exists an adversary where Assumption~\ref{a:5} is necessary for $\regret=o(1)$.
\label{theorem:impossibility}
\end{theorem}
The proof can be found in Appendix~\ref{proof:impossibility}. We design an adversary with a choice over two sequences of utilities against two agents. We show that no policy can have vanishing fairness regret w.r.t. the time horizon under both sequences.

\begin{algorithm}[t]
\caption{\algoname{} policy}
\label{alg:primal_dual_ogaogd}
\begin{algorithmic}[1]
	\begin{footnotesize}
  \Statex \textbf{Require: } $\X$, $\alpha \in \reals_{\geq 0}$, $ \brackets{u_{\star, \min}, u_{\star, \max}}$
\State $\Theta \gets\brackets{-1/{u^\alpha_{\star, \min}},- 1/{u^\alpha_{\star, \max}}}^\I$ \Comment{Initialize the dual (conjugate) subspace}
    \State $\x_1 \in \X$; $\dv_1\in \Theta$; \Comment{Initialize allocation $\x_1$ and dual decision $\dv_1$}
\For{$t \in \T$} 
\State Reveal $\Psi_{t,\alpha}(\dv_t, \x_t) = \parentheses{-F_{\alpha}}^\star(\dv_t) - \dv_t \cdot \vec u_t (\x_t)$ \Comment{Incur reward $\Psi_{t,\alpha}(\dv_t, \underline{\x_t})$ and loss $\Psi_{t,\alpha}(\underline{\dv_t}, \x_t)$}

\State$\vec g_{\X, t} \in \partial_{\x} \Psi_{t,\alpha}(\dv_t, \x_t) = \sum_{i \in \I} \theta_{t,i} \partial_{\x} u_{t,i}$ \Comment{Compute supergradient  $\vec g_{\X, t}$ at $\x_t$ of reward $\Psi_{t,\alpha}(\dv_t, \spacedcdot)$}
\State $\vec g_{\Theta, t} =\nabla_{\vec \theta} \Psi_{t,\alpha}(\dv_t, \x_t) = \parentheses{\parentheses{-\theta_{t,i}}^{-1/\alpha} - \vec u_t (\x_t)}_{i \in \I}$\Comment{Compute gradient  $\vec g_{\Theta, t}$ at $\dv_t$ of loss $\Psi_{t,\alpha}(\spacedcdot, \x_t)$} 
\State $\eta_{\X, t} = {\diam{\X}} / {\sqrt{\sum^t_{s=1} \norm{\vec g_{\X, s}}^2_2}}$; $\eta_{\Theta,t} = {\alpha \umin^{-1-1/\alpha}}/{t}$ \Comment{Compute adaptive learning rates}
\State $\x_{t+1} = \Pi_{\X}\parentheses{\x_t  + \eta_{\X, t} \vec g_{\X, t}}$; $\dv_{t+1} =\Pi_{\Theta} \parentheses{\dv_t  - \eta_{\Theta, t} \vec g_{\Theta, t}}$ \Comment{Compute a new  allocation and dual decision}
\EndFor
\end{footnotesize}
\end{algorithmic}
\end{algorithm}

\subsection{\algoname{} Policy}
\label{s:algorithm_description}
Our policy employs a convex-concave function, composed of a convex conjugate term that tracks the global fairness metric in a dual (conjugate) space, and a weighted sum of utilities term that tracks the appropriate allocations in the primal space. This function is used by the policy to compute a gradient and a supergradient to adapt its internal state. In detail, we define the function $\Psi_\alpha: \Theta \times \X \to \reals $  given by  \begin{align}
    \Psi_{t,\alpha} (\dv, \x) \triangleq \parentheses{-F_\alpha}^\star(\vec \theta) - \vec \theta \cdot \vec u_t(\x),
\end{align} 
where $\Theta= \brackets{-1/{u^\alpha_{\star, \min}},- 1/{u^\alpha_{\star, \max}}}^\I \subset \reals_{<0}^\I$ is a subspace of the dual (conjugate) space, and $\parentheses{-F_\alpha}^\star$ is the \emph{convex conjugate} (see Definition~\ref{def:conjugate} in Appendix) of $-F_{\alpha}$ given by for any $\dv \in \Theta$
\begin{align}
   \parentheses{-F_{\alpha}}^\star(\vec \theta) &= \begin{cases}
    \sum_{i \in \I} \frac{\alpha(-\theta_i)^{1-1/\alpha} - 1}{1-\alpha} &  \text{ for $\alpha \in \reals_{\geq 0} \setminus \set{1}$}, \\ 
      \sum_{i \in \I}  - \log(-\theta_i) - 1   &\text{ for $\alpha  = 1$}.
    \end{cases}
\end{align}
The policy is summarized in Algorithm~\ref{alg:primal_dual_ogaogd}.  The algorithm only requires as input: the set of eligible allocations $\X$, the $\alpha$-fairness parameter in $\reals^\I_{\geq 0}$, and the range $\brackets{u_{\star, \min}, u_{\star, \max}}$ of  values of the  average utility obtained by the optimal allocation~\eqref{eq:hf_objective}, i.e., $\frac{1}{T} \sum_{t \in \T}\u_t(\x_{\star}) \in \brackets{u_{\star, \min}, u_{\star, \max}}^{\I} \subset \reals^{\I}_{>0}$. We stress that the target time horizon $T$ is \emph{not} an input to the policy. 
\new{The utility bounds $u^\alpha_{\star, \min}$ and $u^\alpha_{\star, \max}$ depend on the specific application. For example, 
for the virtualized caching system considered in Section~\ref{s:experiments}, one could simply pick a small enough $\epsilon >0$ as $u^\alpha_{\star, \min}$, and the maximum batch size weighted by the largest retrieval cost in the network as  $u^\alpha_{\star, \max}$  (see Eq.~\eqref{e:utility_expression}). However, if prior information is available to tighten this range, the performance of the algorithm is ameliorated, as reflected in the regret bound in Eq.~\eqref{e:th:u1}.}   

The policy uses its input to initialize the dual (conjugate) subspace $\Theta = \brackets{-1/{u^\alpha_{\star, \min}},- 1/{u^\alpha_{\star, \max}}}^\I$, an allocation $\x_1 \in \X$, and a dual decision $\dv_1 \in \Theta$ (lines~1--2 in Algorithm~\ref{alg:primal_dual_ogaogd}). At a given timeslot $t \in \T$, the allocation $\x_t$ is selected; then a vector-valued utility $\vec u_t (\spacedcdot)$ is revealed and in turn  $\Psi_{t, \alpha} (\spacedcdot, \spacedcdot)$ is  revealed to the policy (line 4 in Algorithm~\ref{alg:primal_dual_ogaogd}). The supergradient $\vec g_{\X,t}$ of $\Psi_{t, \alpha} (\dv_t, \spacedcdot)$ at point $\x_t \in \X$, and the gradient $\vec g_{\Theta, t}$  of $\Psi_{t, \alpha} (\spacedcdot, \x_t)$ at point $\dv_t \in \Theta$ are computed (lines~5--6 in Algorithm~\ref{alg:primal_dual_ogaogd}). The policy then finally performs an adaptation of its state variables $(\x_t, \dv_t)$ through a descent step in the dual space and an ascent step in the primal space through online gradient descent (OGD) and online gradient ascent (OGA) policies,\footnote{Note that a different OCO policy can be used  as long as it has a no-regret guarantee, e.g., online mirror descent (OMD), follow the regularized leader (FTRL), or follow the perturbed leader (FTPL)~\cite{Hazanoco2016, mcmahan2017survey}; moreover, one could even incorporate optimistic versions of such policies~\cite{rakhlin2013online}, to improve the regret rates when the controller has access to accurate predictions.} respectively (line 8 in Algorithm~\ref{alg:primal_dual_ogaogd}).  The learning rates (step size) used are  ``self-confident''~\cite{AUER200248} as they depend on the experienced gradients. Such a learning rate schedule is compelling because it can adapt to the adversary and provides tighter regret guarantees for ``easy'' utility sequences; moreover, it allows attaining an \emph{anytime} regret guarantee, i.e., a guarantee holding for any time horizon $T$.  
In particular, \algoname{} policy in Algorithm~\ref{alg:primal_dual_ogaogd} enjoys the following fairness regret guarantee.
\begin{theorem}
Under assumptions~\ref{a:1}--\ref{a:5}, \algoname{} policy in Algorithm~\ref{alg:primal_dual_ogaogd} attains the following fairness regret guarantee:
\begin{align}
   \regret &\le  \sup_{ \set{\u_t}^T_{t=1} \in {{\U^T}}} \set{{\frac{1.5\diam{\X}}{T}\sqrt{\sum_{t \in \T}\! \norm{\vec g_{\X, t}}^2_2}}\!  +\! \sum^T_{t=1}\!\frac{\alpha  \norm{\vec g_{\Theta, t}}^2_2}{2 u_{\star,\min}^{1 + \frac{1}{\alpha }}  T t}}+\mathcal{O}\parentheses{ \frac{\min\set{\VT, \WT}}{T}} \label{e:th:u1}\\
      &\leq   \frac{1.5\diam{\X} L_{\X}}{{u^{\alpha}_{\star, \min}} \sqrt{T}} + \frac{\alpha L^2_{\Theta} (\log(T) + 1)}{u_{\star,\min}^{1 + \frac{1}{\alpha }} T} +\BigO{\frac{\min{\set{\VT, \WT}}}{T} } \label{e:th:u2}\\
      &=    \BigO{\frac{1}{\sqrt{T}} + \frac{\min\set{\VT, \WT}}{T}} = o(1).\label{e:th:u3}
    \end{align}
\label{th:maintheorem}
\end{theorem}
The proof is provided in Appendix~\ref{proof:t:maintheorem}. We prove that the fairness regret can be upper bounded with the time-averaged regrets of the primal policy operating over the set $\X$ and the dual policy operating over the set $\Theta$, combined with an extra term that is upper bounded with $\min\set{\VT, \WT}$. Note that the fairness regret upper bound in Eq.~\eqref{e:th:u1} can be much tighter than the one in Eq.~\eqref{e:th:u2}, because the gradients' norms can be smaller than their upper bound at a given timeslot $t \in \T$. 
Thanks to its ``self-confident'' learning schedule~\cite{AUER200248}, which dynamically adapts to the observed utilities, our Algorithm~\ref{alg:primal_dual_ogaogd} enjoys an
\emph{any-time} regret guarantee, i.e., it does not require  the knowledge of the target time horizon $T$.

The result in Theorem~\ref{th:maintheorem} is tight, in the sense that 
no policy can have a fairness regret~\eqref{e:b_regret} with better dependency on the time horizon $T$.
Formally, 
\begin{theorem}
\label{theorem:lowerbound}
Any policy $\vec\A$ incurs $ \regret = \Omega\parentheses{\frac{1}{\sqrt{T}}}$ fairness regret~\eqref{e:b_regret} for~$\alpha \geq 0$. 
\end{theorem}
The proof can be found in Appendix~\ref{proof:lowerbound}. We show that the lower bound on regret in online convex optimization~\cite{Hazanoco2016} can be transferred to the fairness regret. 

\new{We discuss in Appendix~\ref{appdendix:time_complexity}, the time-complexity of Algorithm~\ref{alg:primal_dual_ogaogd}  in the context of virtualized caching system application, presented in Section~\ref{s:experiments}.}



\subsection{Adversarial Examples}
\label{s:examples}
In this section, we provide examples of adversaries satisfying Assumptions~\ref{a:1}--\ref{a:5}, with either $\VT = o(T)$ or $\WT = o(T)$, and of stochastic adversaries.

\begin{example} (Adversaries satisfying $\VT = o(T)$)
Consider an adversary selecting utilities such that \begin{align}
\u_t(\x)  =  \u(\x)  +\vec \gamma_t\odot\vec p_t(\x),\end{align} 
where $\u: \X \to \reals^{\I}$ is a fixed utility, the time-dependent function  $\vec p_t : \X \to \reals^{\I}$ is an adversarially selected perturbation with  $\norm{\vec p_t}_{\infty} < \infty $,  $\vec \gamma_t \in \reals^{\I}$ quantifies the severity of the perturbations, and $\vec \gamma_t \odot \vec p_t(\x) = \parentheses{\gamma_{t,i} p_{t,i}(\x)}_{i \in \I}$ is the Hadamard product.  The severity of the perturbations grows sublinearly in time $T$, i.e., $\sum^T_{t=1} \gamma_{t,i}  = o(T)$ for every $i \in \I$. It is easy to check that, in this setting, it holds $\VT = o (T)$.

We provide a simple-yet-illustrative example of such an adversary. We take $\X = [0,1] \subset \reals$, two agents $\I = \set{1,2}$,  fixed utilities $\u (x) = \parentheses{1 - x^2, 1+x}$, adversarial perturbations $\vec p_t(x) =\parentheses{a_{i,t}\cdot x}_{i \in \I}$ where $\vec a_t$ is selected uniformly at random from $[-1,1]^{\I}$ for every $t \in \T$. The perturbations' severity is selected as $\gamma_{\xi_{t,i},i} = t^{-s}$ where $\vec \xi_i: \T \to \T$ is a random permutation of the elements of $\T$ for $i \in \I$.  The performance of Algorithm~\ref{alg:primal_dual_ogaogd} is provided in Fig.~\ref{fig:example1}.  We observe that for larger values of $s$, corresponding to lower perturbation's severity, the policy provides faster the same utilities as  the HF benchmark~\eqref{eq:hf_objective}.
\end{example}

\begin{figure}[t]
    \centering
    \subcaptionbox{$s =\frac{1}{100} $}{\includegraphics[width=.22\linewidth]{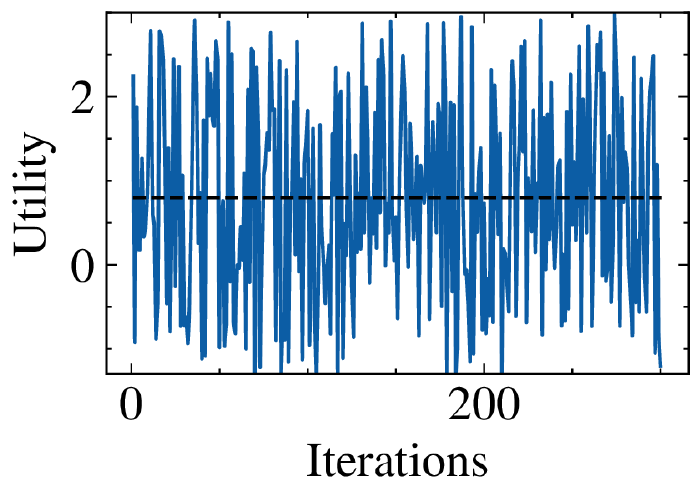}}
    \subcaptionbox{$s = \frac{1}{10}$}{\includegraphics[width=.22\linewidth]{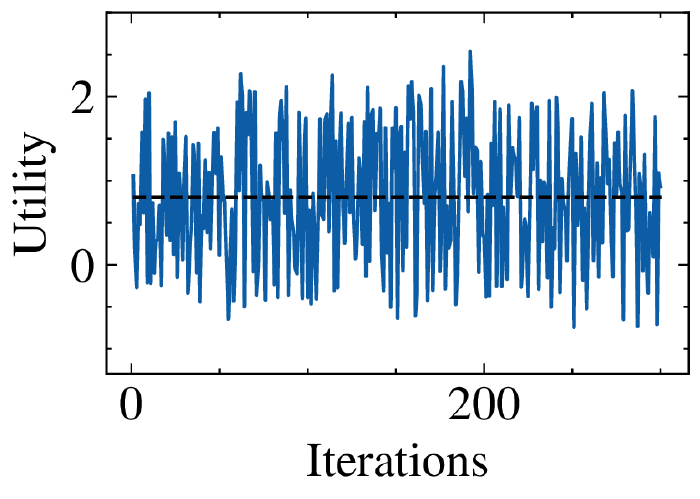}}
    \subcaptionbox{$s = \frac{1}{2}$}{\includegraphics[width=.22\linewidth]{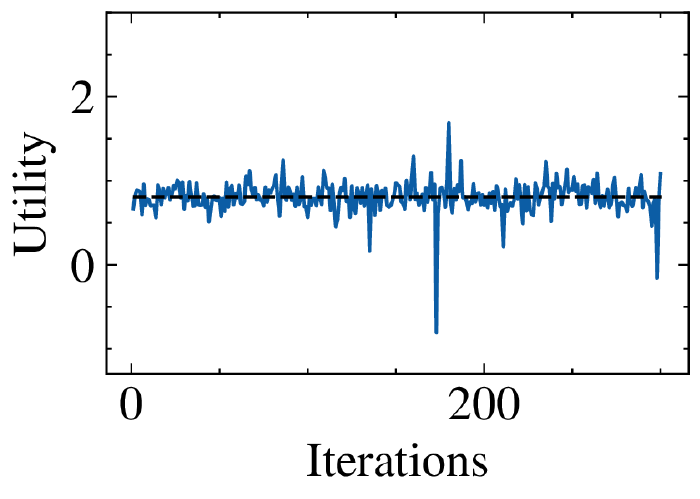}}    
    \subcaptionbox{Time-averaged utility}{\includegraphics[width=.32\linewidth]{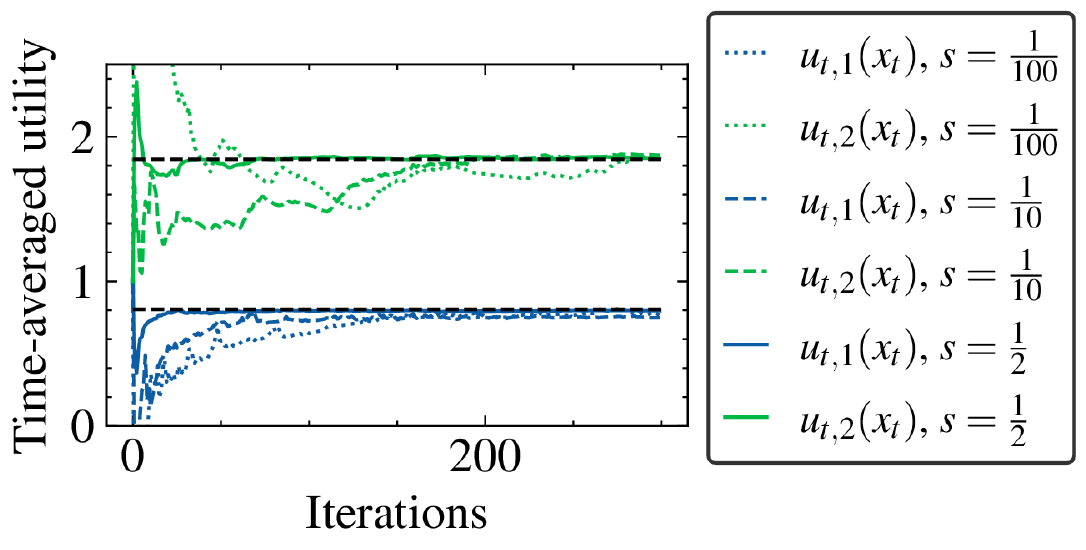}}
    \caption{Subfigures~(a)--(c) provide the  utilities of agent 2 for different values of perturbations' severity parameter $s \in \set{\frac{1}{100}, \frac{1}{10}, \frac{1}{2}}$ under the benchmark's allocation $x_{\star}$. Subfigure~(d) provides the time-averaged utility of two agents. The dark dashed lines represent the utilities obtained by HF objective~\eqref{eq:hf_objective}.}
    \label{fig:example1}
\end{figure}

\begin{example} (Adversaries satisfying $\WT = o(T)$)
\label{example:periodic}
Consider a multiset $\M_t$ of utilities and an adversary that selects a utility $\vec u_t:\X\to \reals^{\I}$ from it. The multiset is updated as follows: if $\M_t \setminus \set{\u_t}\neq \emptyset$, $\M_{t+1} = \M_t \setminus \set{\u_t}$, otherwise, $\M_t = \M_1$. In words, the adversary selects irrevocably elements (utilities) from the set $\M_1$, and, when all the elements are selected, the replenished $\M_1$ is offered  again to the adversary. Consider, without loss of generality, a time horizon $T$  divisible by $\card{\M_1}$ and  the following decomposition for the period $\T$:  $\set{1, 2, \dots, \card{\M_1}}\cup \set{\card{\M_1}+1, \card{\M_1}+2, \dots, 2 \card{\M_1}}\cup \dots = \T_1\cup \T_2\cup\dots\cup \T_{{T}/{\card{\M_1}}}$.  By construction, it holds for every $\x \in \X$
\begin{align}
     \sum_{i \in \I}   \abs{\sum_{t \in \T_k} \delta_{t,i}(\x)} =  0, \;\; \forall k \in \set{1,2,\dots, {T}/{\card{\M_1}}},\label{e:adv_eg_a1}
\end{align}
because when the multiset is fully consumed by the adversary, the average experienced utility is a fixed function.  When $\card{\M_1} = \Theta\parentheses{T^{\epsilon}}$ for $\epsilon  \in [0,1/2)$ it holds  $\sum^{T/\card{\M_1}}_{k=1} \frac{\card{\T_k}^2}{\sum_{k' < k } \card{\T_k}+1} = \BigO{T^{2\epsilon}}$ (see Eq.~\eqref{e:uniform-partitioned-severity}); thus, combined with Eq.~\eqref{e:adv_eg_a1} it holds $\WT = o(T)$. We provide a simple example of such an adversary. Consider $\X = [-1,1]$,  two agents $\I = \set{1,2}$, and the initial  
multiset 
\begin{align}
    \M_1 = \{\vGroup{(1\!-\!x, 1\!-\!(1\!-\!x)^2)}{\text{repeated $10$ times}}, \vGroup{({ 1\!-\!(1\!-\!x)^2, 1\!-\!4x})}{\text{repeated $20$ times}}, \vGroup{({ 1, -2x})}{\text{repeated $10$ times}}\}.\label{e:m1}
\end{align}
We have $\card{\M_1} =40$ and hence $\WT = o(T)$. The performance of Algorithm~\ref{alg:primal_dual_ogaogd} is provided in Fig.~\ref{fig:example2} under different choice patterns over $\M_1$. We observe that the cyclic choice of utilities is more harmful than the u.a.r. one as it leads to slower convergence. Nonetheless, under both settings,  the policy asymptotically yields  the same utilities as the HF benchmark~\eqref{eq:hf_objective}.
\end{example}
\begin{figure}[t]
    \centering
               \subcaptionbox{Allocations (cyclic)}{\includegraphics[width=.24\linewidth]{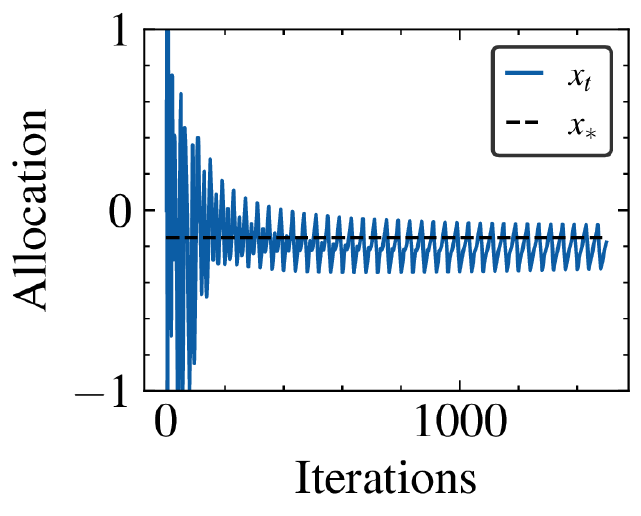}}
       \subcaptionbox{Allocations (u.a.r.)}{\includegraphics[width=.24\linewidth]{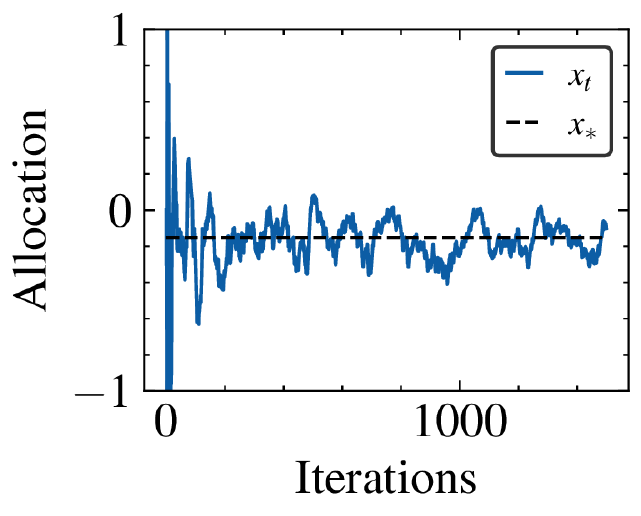}}
\subcaptionbox{Time-averaged utilities (cyclic)}{\includegraphics[width=.23\linewidth]{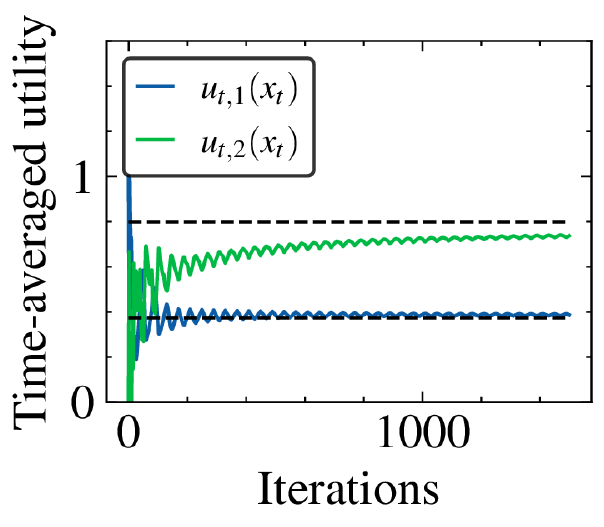}}
\subcaptionbox{Time-averaged utilities (u.a.r.)}{\includegraphics[width=.23\linewidth]{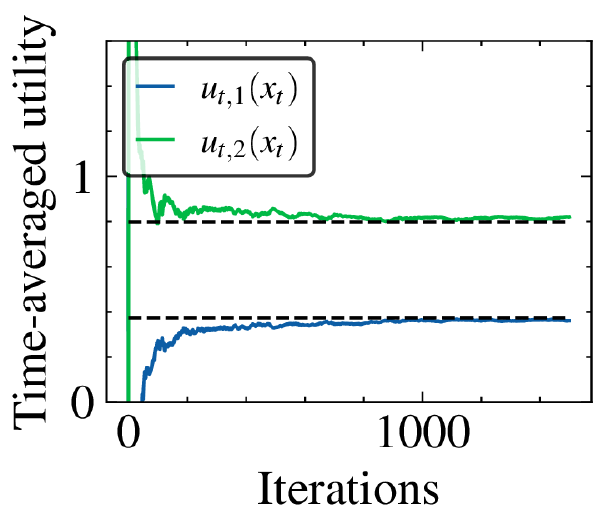}}
    \caption{Subfigures~(a)--(b) provide the allocations of different agents of cyclic and u.a.r. choice of utilities over the set $\M_1$, respectively. Subfigures~(c)--(d) provide the time-averaged utility of cyclic and u.a.r. choice of utilities over the set $\M_1$, respectively. }
    \label{fig:example2}
\end{figure}
\begin{example} (Stochastic Adversary)
\label{e:example3}
Consider a scenario where  $u_{t,i}:\X \to \reals$ are drawn i.i.d. from an unknown distribution $\mathcal{D}_i$. Formally, the following corollary is obtained from Theorem~\ref{th:maintheorem}.

\begin{corollary}
When the utilities $ u_{t,i}:\X \to \reals$ are drawn i.i.d. from an unknown distribution~$\mathcal{D}_i$ satisfying Assumptions~\ref{a:1}--\ref{a:4}, the policy \algoname{} in Algorithm~\ref{alg:primal_dual_ogaogd}  attains the following expected fairness  regret guarantee:
\begin{align}
   \barregret \triangleq\sup_{ {\mathcal{D}_i}, \,i \in \I} \set{ \underset{\substack{u_{t,i} \sim \mathcal{D}_i \\ i \in \I, \,t \in \T}}{\mathbb E} \brackets{\max_{\x \in \X}  F_{\alpha}\left({\frac{1}{T}\sum_{t \in \T}\vec u_{t}(\x_t)}\right) -F_{\alpha}\left({\frac{1}{T}\sum_{t \in \T}\vec u_{t}(\x_t)}\right)}} = \BigO{\frac{1}{\sqrt{T}}}.
\end{align}
Moreover, it holds with probability one: $\regret  \leq  0$ for $T \to \infty$.
\label{corollary:stochastic}
\end{corollary}
The proof is 
in Appendix~\ref{proof:stochastic}. The expected fairness regret guarantee follows from Theorem~\ref{th:maintheorem} and observing that $\mathbb E \brackets{\vec \delta_t (\x)} = \vec 0$ 
for any $t \in \T$ and $\x \in \X$. The high probability fairness regret guarantee for large~$T$
is obtained through Hoeffding's inequality paired with Eq.~\eqref{e:adv2}.
\end{example}

Note that we provide additional examples of adversaries, in the context of the application of our policy to a virtualized caching system, in Section~\ref{s:experiments}.

\section{Extensions}
\label{s:extensions}

In this section, we first show that our algorithmic framework extends to cooperative bargaining settings, in particular Nash bargaining~\cite{Nash1950}. Secondly, we show that our framework also extends to the weighted $\alpha$-fairness criterion.
\subsection{Nash Bargaining}

Nash bargaining solution (NBS), proposed in the seminal paper~\cite{Nash1950}, is a fairness criterion for dispersing to a set of  agents the utility of their cooperation. 
The solution guarantees that, whenever the agents cooperate,
each agent achieves an individual  performance that exceeds its performance when operating independently.
This latter is also known as the disagreement point. NBS comes from the area of cooperative game theory, and it is self enforcing, i.e., the agents will agree to apply this solution without the need for an external authority to enforce compliance. NBS has been extensively applied in communication networks, e.g., to transmission power control~\cite{Boche2011},  mobile Internet sharing among wireless users~\cite{Iosifidis2017},  content delivery in ISP-CDN partnerships~\cite{Wenjie2009}, and  cooperative caching in information-centric networks~\cite{Liang2017}. 

Nash bargaining can be incorporated through our fairness framework when $\alpha = 1$, and utilities as redefined for every $t \in \T$ as follows $ \u_t' (\x) = \u_t(\x) - \u^d_t$ where $u^d_i$ is the disagreement point of agent $i \in \I$.   In particular, \algoname{} provides the same guarantees. We  also note that the dynamic model generalizes the NBS solution by allowing both the utilities and the disagreement points to change over time, while the benchmark is defined using~\eqref{eq:hf_objective} and $\alpha=1$.  Hence, the proposed \algoname{} allows the agents to collaborate without knowing in advance the benefits of their cooperation nor their disagreement points, in a way that guarantees they will achieve the commonly agreed NBS at the end of the horizon T (asymptotically).
\subsection{The $(\vec w, \alpha)$-Fairness}
The weighted $\alpha$-fairness or simply $(\vec w, \alpha)$-fairness with  $\alpha \geq 0$  and $\vec w \in \Delta_{\I}\subset \reals_{\geq 0}$, where $\Delta_{\I}$ is the probability simplex with support $\I$,  is defined as~\cite{mo2000fair}:
\begin{definition}
A $(\vec w, \alpha)$-fairness  function $F_{\vec w, \alpha}:\U\to \reals$ is parameterized by the inequality aversion parameter $\alpha \in \reals_{\geq 0}$, weights $\vec w \in \Delta_{\I}$ and it is given by  $ F_{\vec w, \alpha}(\u) \triangleq \sum_{i \in \I} w_i f_{\alpha}(u_i)$ for every $\u \in \U$. Note that $\U \subset \reals^\I_{\geq 0}$ for $\alpha <1$, and $\U \subset \reals^\I_{>0}$ for $\alpha \geq 1$.
\end{definition}
It is easy to check that our $\alpha$-fairness framework captures the $(\vec w, \alpha)$-fairness by simply redefining the utilities incurred at time $t \in \I$ for agent $i \in \I$ as follows: $u'_{t,i}(\x) = w^{\frac{1}{1-\alpha}}_i u_{t,i}(\x)$ for $\alpha \in \reals_{\geq 0} \setminus \set{1}$, otherwise $u'_{t,i}(\x) = \parentheses{u_{t,i}(\x)}^{w_i}$.
Note that for $\alpha = 1$ and uniform weights, we recover the Nash bargaining setting discussed previously; otherwise, we recover asymmetric Nash bargaining in which the different weights correspond to the bargaining powers of players~\cite{harsanyi1972generalized}.

\section{Application}
\label{s:experiments}

In order to demonstrate the applicability of the proposed fairness framework, we target a  representative resource management problem in virtualized caching systems where different caches cooperate by serving jointly the received content requests. This problem has been studied extensively in its static version, where the request rates for each content file are a priori known and the goal is to decide which files to store at each cache to maximize a fairness metric of cache hits across different caches, see for instance \cite{Liang2017, LIU2020102138}. We study the more realistic version of the problem where the request patterns are unknown. This online caching model has been recently studied as a learning problem in a series of papers \cite{paschos2019learning, sisalem2021no, mhaisen2022online, paria2021texttt,bura2021learning,Li2021}, yet none of them handles fairness metrics.

\subsection{Multi-Agent Cache Networks}
\paragraph{Cache network.} 
We assume that time is slotted and the set of timeslots is denoted by $\T \triangleq \set{1, 2, \dots, T}$. We consider a catalog of equally-sized files $\F \triangleq\set{1, 2, \dots, F}$.\footnote{Note that we assume equally-sized files to streamline the presentation. Our model supports unequally-sized files by replacing the cardinality constraint in Eq.~\eqref{e:constraint} with a knapsack constraint and the set $\X_c$ (defined in~\eqref{e:constraint}) remains convex.} We model a cache network at timeslot $t \in \T$ as an undirected weighted graph $G_t(\C, \E)$, where $\C \triangleq\set{1,2,\dots ,C}$ is the set of caches, and $(c, c') \in \E$ denotes the link connecting cache $c$ to $c'$ with associated weight $w_{t, (c,c')} \in \reals_{>0}$.  Let $\P_{t, (c,c')} =\set{c_1, c_2,\dots, c_{\card{\P_{t, (c,c')}}}}  \in \C ^{\card{\P_{t, (c,c')}}}$ be the shortest path at timeslot $t \in \T$ from cache $c$ to cache $c'$ with associated weight  $    w^{\mathrm{sp}}_{t, (c,c')} \triangleq \sum^{\card{\P_{t, (c,c')}}-1}_{k =1} w_{t, (c_k, c_{k+1})}$.

We assume for each file $f \in\F$ is permanently stored at a set $\Lambda_f (\C) \subset \C$ of designated repository  servers.
\new{Moreover, each cache can store fractions of the file and fractions of the same file at different caches can be additively combined.}\footnote{\new{
This is a common assumption~\cite{shanmugam2013femtocaching,paschos2019learning}, which models situations where each file can be split in a large number of small chunks and each cache can store random linear combinations of such chunks.
Guarantees for this fractional setting can be readily transferred to an integral setting through randomized rounding techniques~\cite{paria2021texttt,geocaching2015, Ioannidis16, sisalem21arxiv}.}} 
We denote by $x_{t,c,f} \in [0,1]$ the fraction
of file $f \in\F$ stored at cache $c \in \C$ at timeslot $t \in \T$. The state of cache $c \in \C$ is given by $\x_{t,c}$  drawn from the set
\begin{align}
    \X_c \triangleq \set{\x \in \brackets{0, 1}^
    \F: \sum_{f \in\F} x_f \leq  k_c, x_f \geq \mathds{1}{\parentheses{c \in \Lambda_f (\C)}}, \forall f \in\F},\label{e:constraint}
\end{align}
where  $k_c \in \naturals$ is the capacity of cache $c \in \C$, and $\mathds{1}{\parentheses{\chi}} \in \set{0,1}$ is the indicator function set to 1 when condition $\chi$ is true. Thus, the state of the cache network belongs to  $\X \triangleq \bigtimes_{c \in \C} \X_c$.  The system model is summarized in Fig.~\ref{fig:system_model}, and it is aligned with many recent papers focusing on learning for caching~\cite{Ioannidis16,paschos2019learning,paria2021texttt,sisalem21arxiv}.
\begin{figure}[t]
    \centering
    \includegraphics[width=.7\linewidth]{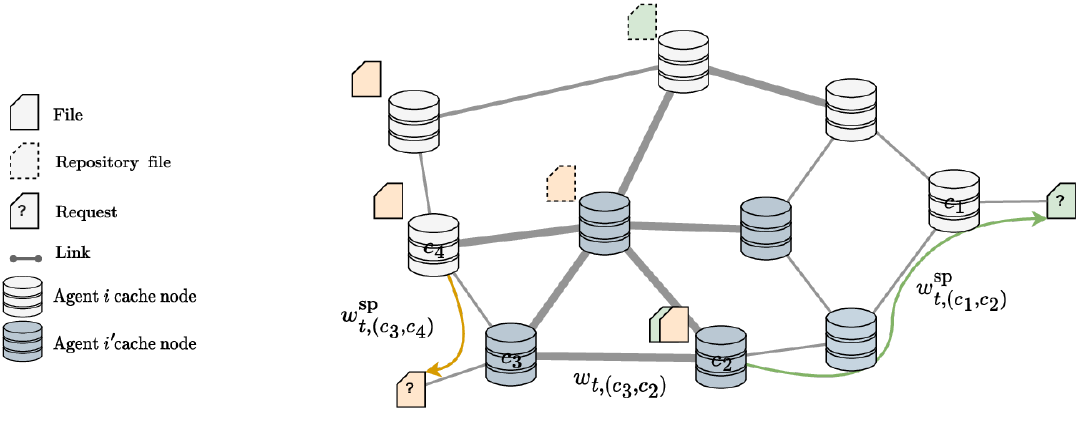}
    \caption{System model: a network comprised of a set of caching nodes $\C$. A request arrives at a cache node $c \in \C$, it can be partially served locally, and if needed, forwarded along the shortest retrieval path to another node to retrieve the remaining part of the file; a utility is incurred by the cache owner $i \in \I$. A set of permanently allocated files are spread across the network guaranteeing requests can always be served.
    }
    \label{fig:system_model}
\end{figure}
\paragraph{Requests.}
We denote by $r_{t,c,f} \in \naturals \cup \set{0}$ the number of requests for file $f \in \F$ submitted by users  associated to  cache $c \in \C$, during slot $t \in \T$. The request batch arriving at timeslot $t \in \T$ is denoted by $\vec r_{t} = \parentheses{r_{t,c,f}}_{(c,f) \in \C \times \F}$ and belongs to the set $$    \mathcal{R}_t \triangleq \set{\vec r \in \parentheses{\naturals \cup \set{0}}^{ \C \times \F}: \sum_{c \in \C}\sum_{f \in\F} r_{c,f} \leq R_t},$$ where $R_t \in \naturals$ is the total number of requests (potentially) arriving at the system at timeslot $t \in \T$. 
\paragraph{Caching gain.}
We consider an agent $i \in \I$  holds a set of caches $ \Gamma_i (\C) \subset \C$, and $\dot\bigcup_{i \in \I}  \Gamma_i (\C)  =  \C$. Hence, the allocation set of agent $i$ is given by $\X_i = \bigtimes_{c \in \Gamma_i(\C)} \X_c$. Requests arriving at cache $c\in \C$ can be partially served locally, and if needed,  forwarded along the shortest path to a nearby cache $c'\in \C$ storing the file, incurring a retrieval cost $w^{\mathrm{sp}}_{t, (c,c')}$. Let $\phi_{t,i, c} \triangleq \arg\min_{c' \in \Lambda_i (\C)} \set{w^{\mathrm{sp}}_{t, (c, c')}}$ and $\Phi_{t,i, c}: \set{1, 2, \dots, \phi_{t, i, c}} \subset \C \to\C$ be a map providing a retrieval cost ordering for every $c \in \set{1,2, \dots, \phi_{t,i, c}}$, $t \in \T$, and $i \in \I$,   i.e.,
\begin{align}
 w^{\mathrm{sp}}_{t, (c, \Phi_{t, i,c}(\phi_{t,i, c}))} = \min\set{w^{\mathrm{sp}}_{t, (c,c')}: c' \in \Lambda_f(\C)} \geq  \dots \geq    w^{\mathrm{sp}}_{t, (c, \Phi_{t, i,c}(2)))} \geq  w^{\mathrm{sp}}_{t, (c, \Phi_{t, i,c}(1)))} = 0. 
\end{align} 
When a request batch $\vec r_t \in \mathcal{R}_t$ arrives at timeslot $t \in \T$, agent $i \in \I$ incurs the following cost:
\begin{align*}
    \mathrm{cost}_{t, i}(\x) \triangleq\!\!\!\! \sum_{c \in \Gamma_i(\C)} \sum_{f \in\F} r_{t,c,f}\sum^{\phi_{t,i, c}-1}_{k=1} \parentheses{w^{\mathrm{sp}}_{t, (c, \Phi_{t,i, c} (k+1))} -   w^{\mathrm{sp}}_{t, (c, \Phi_{t,i, c} (k))}} \parentheses{1 - \min\set{1, \sum^k_{k'=1} x_{\Phi_{t,i, c}(k'), f}}}.
\end{align*}
This can be interpreted as a QoS cost paid by a user for the additional delay to retrieve  part of the file from another cache, or it can represent the load on the network to provide the missing file. Note that by construction, the maximum cost is achieved for a network state, where all the caches are empty except for the repository allocations; formally, such state is given by $\x_0 \triangleq \parentheses{\mathds{1}\parentheses{c \in \Lambda_f(\C)}}_{(c,f) \in \C \times \F} \in \X$, and the cost of the agent at this state is given by
\begin{align}
   \mathrm{cost}_{t, i}(\x_0) &=  \sum_{c \in \Gamma_i(\C)} \sum_{f \in\F} r_{t,c,f} \min\set{w^{\mathrm{sp}}_{t, (c, c')}: c' \in \Lambda_f(\C)} \\&= \sum_{c \in \Gamma_i(\C)} \sum_{f \in\F} r_{t,c,f}\sum^{\phi_{t,i, c}-1}_{k=1} \parentheses{w^{\mathrm{sp}}_{t, (c, \Phi_{t,i, c} (k+1))} -   w^{\mathrm{sp}}_{t, (c, \Phi_{t,i, c} (k))}},
\end{align}
We can define the caching utility at timeslot $t\in \T$ as the cost reduction due to caching as:
\begin{align}
    u_{t, i}(\x) &\triangleq \mathrm{cost}_{t,i}(\x_0) -  \mathrm{cost}_{t,i}(\x) \\
    &=  \sum_{c \in \Gamma_i(\C)} \sum_{f \in\F} r_{t,c,f}\sum^{\phi_{t,i, c}-1}_{k=1} \parentheses{w^{\mathrm{sp}}_{t, (c, \Phi_{t,i, c} (k+1))} -   w^{\mathrm{sp}}_{t, (c, \Phi_{t,i, c} (k))}} \min\set{1, \sum^k_{k'=1} x_{\Phi_{t,i, c}\parentheses{k'} , f}}.\label{e:utility_expression}
\end{align}
The caching utility is a weighted sum  of concave functions with positive weights, and thus concave in $\x \in \X$. 
It is straightforward to check that this problem always satisfies Assumptions~\ref{a:1}--\ref{a:4}. 
The request batches and the time-varying retrieval costs determine whether Assumption~\ref{a:5} holds. For example, this is the case when request batches are drawn i.i.d. from a fixed unknown distribution (see Example~\ref{e:example3}).

\subsection{Results}



Below we describe the experimental setup\footnote{\new{Our code is publicly available at https://github.com/tareq-si-salem/Online-Multi-Agent-Cache-Networks}} of the multi-agent cache networks problem, the request traces, and competing policies. Our results are summarized as follows: 
\begin{enumerate}
    \item Under stationary requests and small batch sizes (leading to large utility deviations from one timeslot to another),    \algoname{} achieves the same time-averaged utilities as the offline benchmark, whereas \slotalgoname, a counterpart policy to \algoname{} targeting slot-fairness~\eqref{eq:sf_objective}, diverges and is unable to reach the Pareto front.
    \item In the Nash bargaining scenario,  \algoname{} achieves the NBS 
    in all cases, while \slotalgoname{} fails when the disagreement points are exigent, i.e.,  an agent can guarantee itself a high utility.
    \item Widely used \lfu{} and \lru{}
    might perform arbitrarily bad w.r.t. fairness, and not even achieve any point in the Pareto front (hence, they are not only unfair, but also inefficient).
    \item Fairness comes at a higher price when $\alpha$ is increased or the number of agents is increased.  This observation on the price of fairness provides experimental evidence for previous work~\cite{bertsimas2012efficiency}.
      \item \algoname{} is robust to different network topologies and is able to obtain time-averaged utilities that match the offline benchmark.  
    \item Under non-stationary requests, \horizonalgoname{} policy achieves the same time-averaged utilities as the offline benchmark, whereas \slotalgoname{} can perform arbitrarily bad providing allocations that are both unfair and inefficient
\end{enumerate}

\paragraph{General Setup.}
We consider three synthetic network topologies (\SC, \BT, and \Grid), and two real network topologies (\Abilene{} and \GEANT). A visualization of the network topologies is provided in Figure~\ref{fig:topologies}. The specifications of the network topologies used across the experiments  are provided in Table~\ref{table:topologies} in the Appendix.  A repository node permanently stores the entire catalog of files. The retrieval costs along the edges are sampled u.a.r. from $\set{1,2, \dots, 5}$, except for edges directly connected to a repository node which are sampled u.a.r. from $\set{6, 7, \dots, 10}$. All the retrieval costs remain fixed for every $t \in \T$. The capacity of each cache is sampled u.a.r. from $\set{1,2, \dots, 5}$, but for the  \SC{} topology in which each cache has capacity 5. An agent $i \in \I$ has a set of query nodes denoted by $\mathcal{Q}_i \subset \Gamma_i(\C)$, and a query node can generate  a batch of requests from a catalog with $\card{\F} = 20$ files.  Unless otherwise said, we consider $u_{\star, \min}= 0.1$ and $u_{\star, \max} = 1.0$. The fairness benchmark refers to the maximizer of the HF objective~\eqref{eq:hf_objective}, and the utilitarian benchmark refers to the  maximizer of HF objective~\eqref{eq:hf_objective} for $\alpha = 0$. 

\begin{figure}[t]
	\centering
	\subcaptionbox*{}{\includegraphics[trim={0 5.3cm 0 0},clip,width=0.7\linewidth]{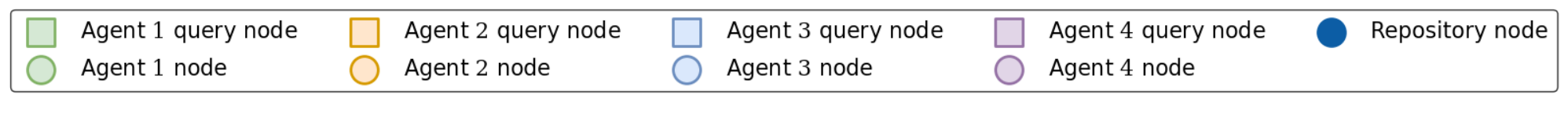}}\vspace{-2em}\\
	\subcaptionbox{\SC\label{sfig:1}}{\includegraphics[width=.13\linewidth]{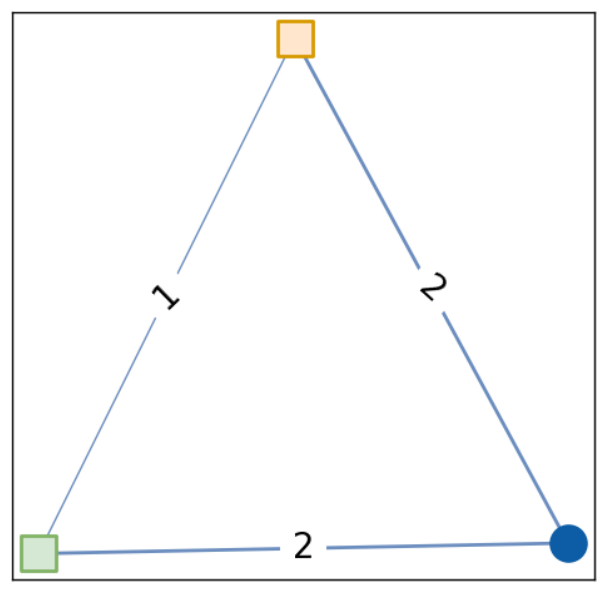}}
	\subcaptionbox{\BT-1\label{sfig:2}}{\includegraphics[width=0.13\linewidth]{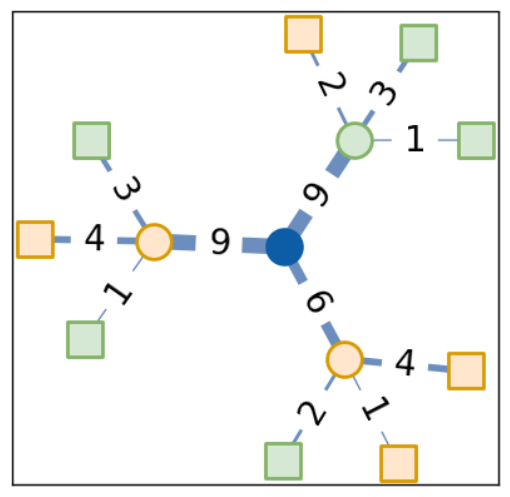}}
	\subcaptionbox{\BT-2\label{sfig:3}}{\includegraphics[width=0.13\linewidth]{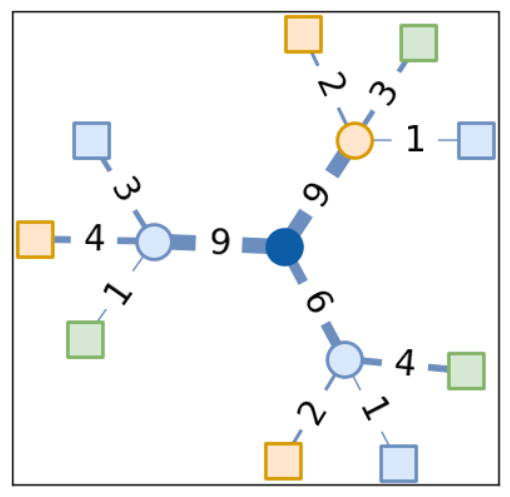}}
	\subcaptionbox{\BT-3\label{sfig:4}}{\includegraphics[width=0.13\linewidth]{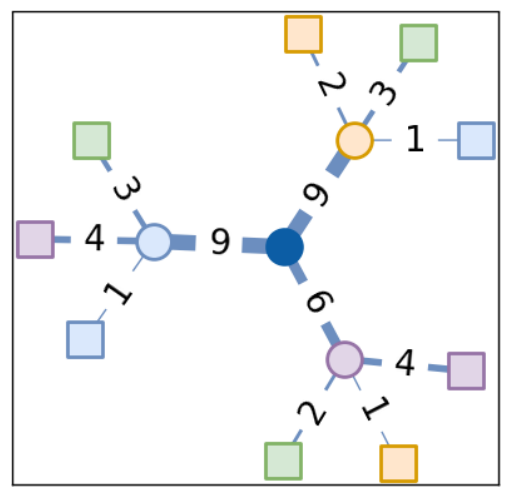}}
	\subcaptionbox{\label{sfig:5}\Grid}{\includegraphics[width=.13\linewidth]{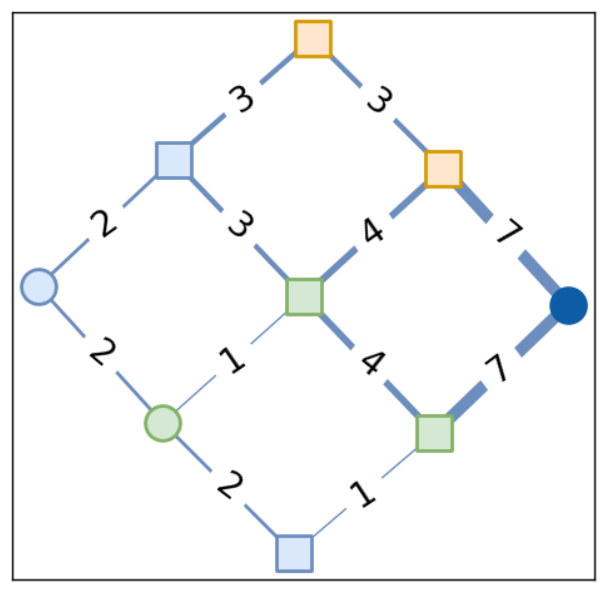}}
	\subcaptionbox{\label{sfig:6}\Abilene}{\includegraphics[width=.13\linewidth]{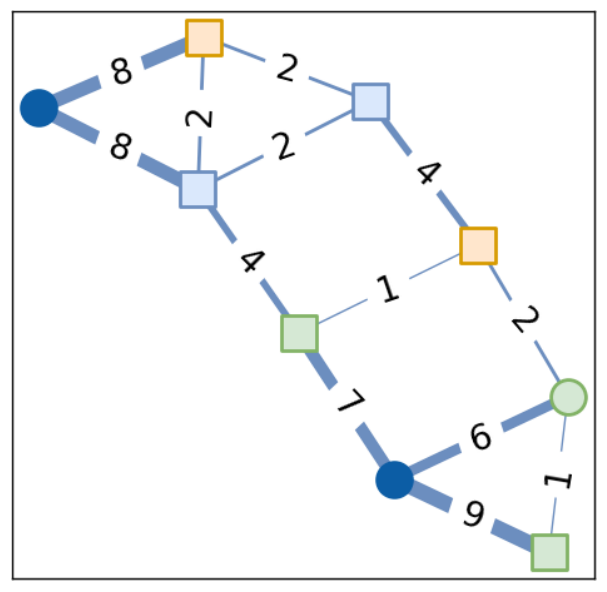}} \subcaptionbox{\label{sfig:7}\GEANT}{\includegraphics[width=.13\linewidth]{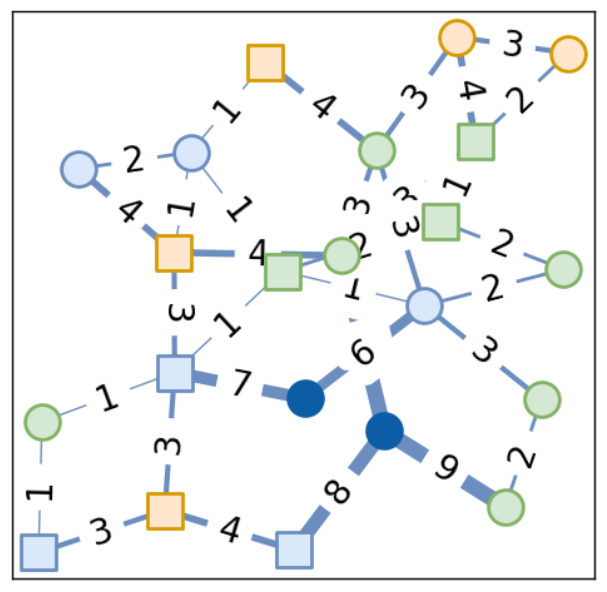}}
	\caption{Network topologies used in experiments.}
	\label{fig:topologies}
\end{figure}

\paragraph{Traces.} Each query node generates requests according to the following:
\begin{itemize}
    \item {Stationary trace (parameters: $\sigma, R, T, F$).} Requests are sampled i.i.d. from a Zipf distribution with exponent $\sigma \in \reals_{\geq 0}$ from  a catalog of files of size $F$. The requests are grouped into  batches of  size $\card{\R_{t}} = R, \forall t \in \T$. 
    \item {Non-Stationary trace (parameters: $\sigma, R, T, F, D$).} Similarly, requests are sampled i.i.d. from a catalog of $F$ files according to a Zipf distribution with exponent $\sigma \in \reals_{\geq 0}$. Every $D$ requests, the popularity distribution is modified in the following fashion:  file $f \in \F = \set{1,2, \dots, F}$ assumes the popularity of file $f' = (f +F/2) \mod F$ ($F$ is even).  The requests are grouped into batches of  size $ \card{\R_{t}} = R, \forall t\in \T$. 
\end{itemize}
\new{The stationary  trace corresponds to the stochastic adversary in Example~\ref{e:example3}, and the non-stationary trace corresponds to a stochastic adversary with perturbations satisfying the partitioned-severity condition in Eq.~\eqref{e:adv2}. } Two sampled traces are depicted in Figure~\ref{fig:traces} in the Appendix. Unless otherwise said, query nodes generate \emph{Stationary} traces
and $\sigma = 1.2$, $T = 10^4$, $R = 50$, and $D = 50$.

\paragraph{Policies.}
We implement the following policies and use them as comparison benchmarks for \algoname.
\begin{itemize}
	\item The classic \lru{} and \lfu{} policies. A request is routed to the cache with minimal retrieval cost among those that store the requested file and this cache provides the content and updates its state corresponding to a hit.
	Moreover, all caches  with a lower retrieval cost update their state as if a miss occurred locally. 
	This corresponds to the popular path replication algorithm~\cite{pathreplication,Ioannidis16}, equipped with \lru{} or \lfu{},  adapted to our setting.
	\item Online slot-fairness (\slotalgoname) policy.
	This policy is the slot-fairness~\eqref{eq:sf_objective} counterpart of \algoname. It is obtained by configuring Algorithm~\ref{alg:primal_dual_ogaogd} with dual (conjugate) subspace $\Theta = \set{(-1)_{i \in \I}}$ 
	(i.e., taking $\alpha \to 0$), which makes ineffective the dual policy in Algorithm~\ref{alg:primal_dual_ogaogd}. The revealed utilities at timeslot $t \in \T$ are the $\alpha$-fairness transformed utilities $\u'_t(\spacedcdot) = (f_{\alpha} \parentheses{u_{t,i} (\spacedcdot)})_{i \in \I}$. The primal allocations are still determined by the same self-confident learning rates' schedule as \algoname{} for a fair comparison.
	The resulting policy is a no-regret policy (see Lemma~\ref{lemma:ogd_regret} in Appendix) w.r.t. the slot-fairness benchmark~\eqref{eq:sf_objective} for some $\alpha \in \reals_{\geq 0}$. 
\end{itemize}

\paragraph{Static analysis of symmetry-breaking parameters.} We start with a numerical investigation of the potential caching gains, and how these are affected by the fairness parameter~$\alpha$. In Figure~\ref{fig:static_exploration},  we consider the \SC{} topology and different values of $\alpha \in [0,2]$. We show the impact  on the fairness benchmark of varying the request patterns ($\sigma \in \set{0.6, 0.8, 1.0, 1.2}$)  for agent~2 under the \emph{Stationary} trace in Fig.~\ref{fig:static_exploration}~(a), and of   varying the retrieval costs between  agent 1's cache and the repository ($w_{(1,3)} \in [2.5,4]$).  
In Figure~\ref{fig:static_exploration}~(a), we observe decreasing the skewness of the popularity distribution decreases the utility of agent~2 as reflected by the downward shift of the Pareto front. We note that, as far as the file popularity distribution at agent 2 is close to the one at agent~1 ($\sigma = 1.2$), different values of alpha still provide similar utilities.  However, in highly asymmetric scenarios, different values of $\alpha$ lead to   clearly distinct utilities for each agent. We also note  that higher values of $\alpha$ guarantees higher fairness by  that increasing the utility of agent~2. Similarly, in Figure~\ref{fig:static_exploration}~(b), we observe increasing the retrieval cost for agent~1 decreases the utility achieved  by the same agent, as reflected by the leftward shift of the Pareto front; moreover, increasing the retrieval costs (higher asymmetry) highlights the difference between different values of~$\alpha$.
\begin{figure}[t]
	\centering
	\subcaptionbox{\label{sfig:static_exploration2}}{\includegraphics[width=0.36\linewidth]{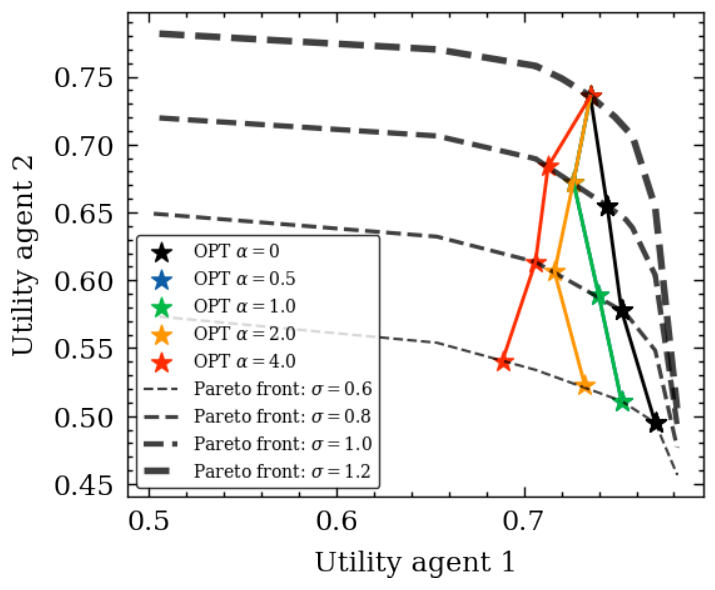}}
	\subcaptionbox{\label{sfig:static_exploration1}}{\includegraphics[width=0.36\linewidth]{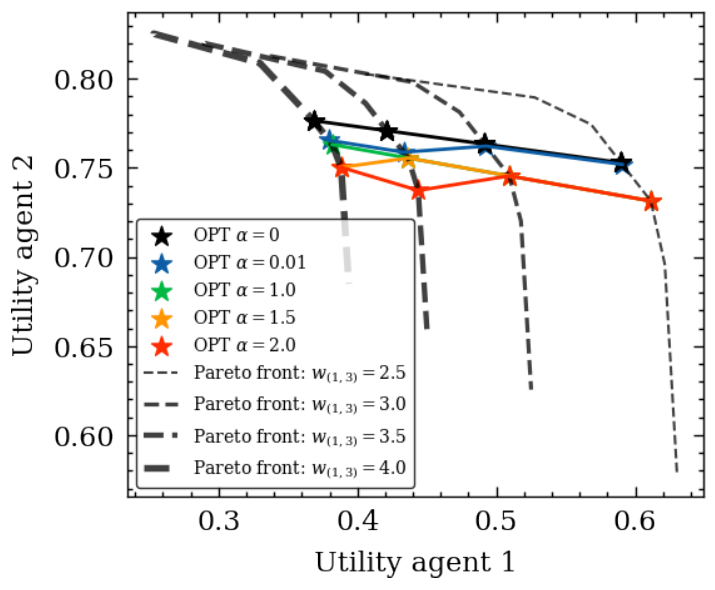}}
	\caption{Pareto front and fairness benchmark's utilities for different values of $\alpha \in [0,2]$ under different request patterns~(a) ($\sigma \in \set{0.6, 0.8, 1.0, 1.2}$) for agent~2, and different retrieval costs~(b) between agent~1’s cache and the repository ($w_{(1,3)} \in [2.5,4.0]$). }
	\label{fig:static_exploration}
\end{figure}

\paragraph{Online analysis of symmetry-breaking parameters.} In Figure~\ref{fig:online_setting1}, we consider the \SC{} topology, and different values of $\alpha \in \set{0,1,2}$. In Figure~\ref{fig:online_setting1}~(a)--(b) we consider the retrieval cost  $w_{(1,3)}= 3.5$ between agent~1's cache node and the repository node. In Figure~\ref{fig:online_setting1}~(c)--(d)  query node of agent~1 generates \emph{Stationary} trace ($\sigma = 1.2$) and query node of agent~2 generates \emph{Stationary} trace ($\sigma = 0.6$). We consider two fixed request batch sizes $R \in \set{1, 50}$. 

In Figures~\ref{fig:online_setting1}~(a) and~(c) (for batch size $R=1$) \horizonalgoname{} approaches the fairness benchmark's utilities for  different values of $\alpha$, but \slotalgoname{} diverges for values of $\alpha \neq 0$. For increased request batch size $R = 50$,   \horizonalgoname{} and \slotalgoname{} exhibit  similar behavior. This is expected under stationary utilities; increasing the batch size reduces the variability in the incurred utilities at every timeslot, and the horizon-fairness and slot-fairness objectives become closer  yielding  similar allocations. Note that this observation implies that \slotalgoname{} is only capable to converge for utilities with low variability, which is far from realistic scenarios. \lfu{} policy outperforms \lru{} and both policies do not approach the Pareto front; thus, the allocations selected by such policies are inefficient and unfair.
\begin{figure}[t]
	\centering
	\subcaptionbox*{}{\includegraphics[trim={0 7.4cm 0 0},clip,width=0.7\linewidth]{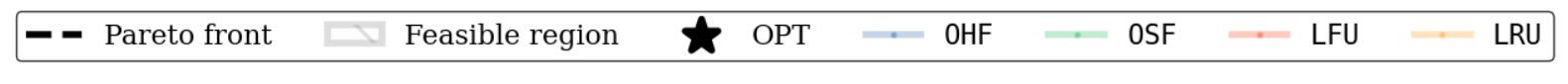}}\vspace{-2em}\\
	\subcaptionbox{$R=1$ \label{sfig:online_setting11}}{\includegraphics[width=0.3\linewidth]{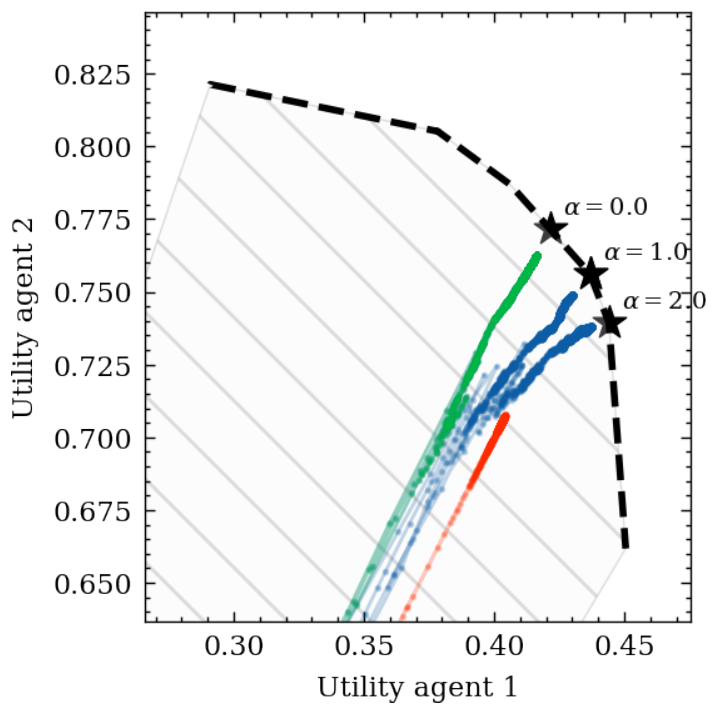}}
	\subcaptionbox{$R=50$\label{sfig:online_setting12}}{\includegraphics[width=0.3\linewidth]{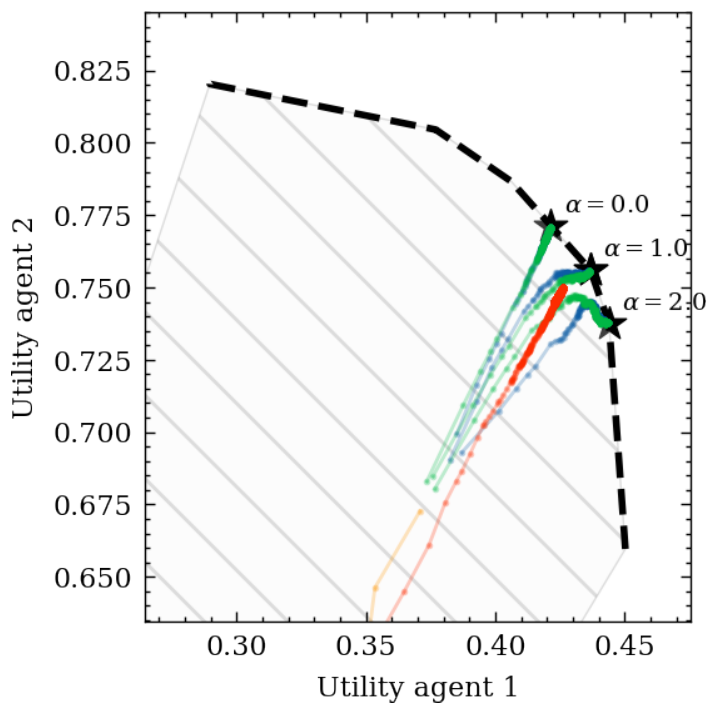}}

	\subcaptionbox{$R=1$\label{sfig:online_setting21}}{\includegraphics[width=0.3\linewidth]{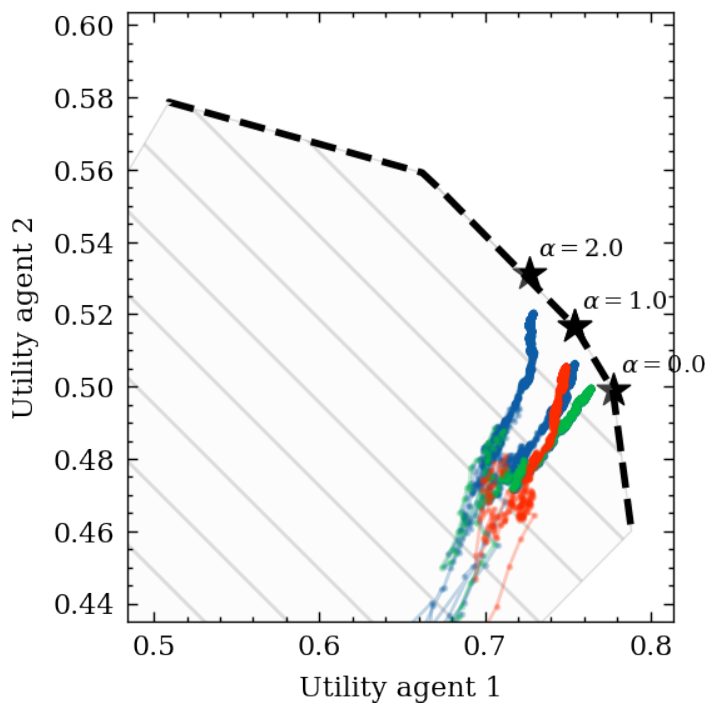}}
	\subcaptionbox{$R=50$\label{sfig:online_setting22}}{\includegraphics[width=0.3\linewidth]{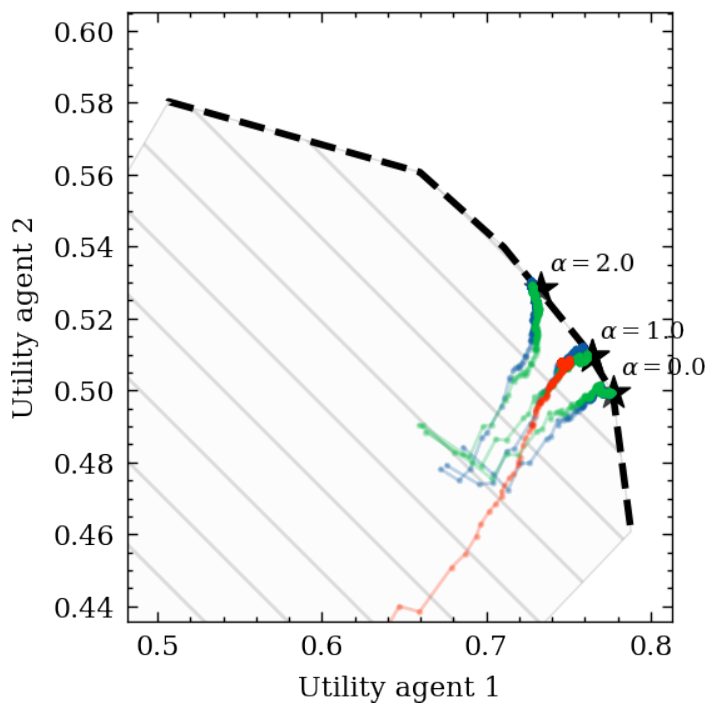}}
	\caption{Time-averaged utilities of policies \horizonalgoname, \slotalgoname, \lru, and \lfu{} under \SC{} topology. Subfigures~(a)--(b) are obtained under  retrieval cost $w_{(1,3)} = 3.5$ for agent 1's query node. Subfigures~(c)--(d) are obtained when agent 2’s query node generates \emph{Stationary trace} ($\sigma = 0.6$). Markers correspond to  iterations in~$\set{100, 200, \dots, 10^4}$. 
	}
	\label{fig:online_setting1}
\end{figure}
\paragraph{Nash bargaining.} 
In Figure~\ref{fig:nash_bargainig}, we consider the \SC{} topology and $\alpha=1$. 
We select different disagreement utilities for agent~2 in $\set{0.0,0.5,0.7,0.75}$, i.e., different utility values agent~2 expects to guarantee itself even in the absence of cooperation. Note how higher values of disagreement utilities 
lead to higher utilities for agent~2 at the fairness benchmark. We select $u_{\star,\min} = 0.01$.

For a small batch size ($R = 1$), \horizonalgoname{} approaches the same utilities achieved by the fairness benchmark for different disagreement points, whereas \slotalgoname{} fails to approach the Pareto front. Similarly, for a larger batch size $R = 50$,    \horizonalgoname{} approaches the fairness benchmark for different disagreement points, but the Pareto front is reached faster than with a batch size $R=1$. \slotalgoname{} diverges for non-zero disagreement points 
when $R=50$, because 
the allocation selected 
for some agent~$i\in\I$ can be smaller than its disagreement utility (i.e., $u_{t,i}(\x_t) - u_{t,i}< 0$), while the $\alpha$-fairness function is only defined for positive arguments.

\begin{figure}[t]
	\centering
	 	\subcaptionbox*{}{\includegraphics[trim={0 7.4cm 0 0},clip,width=0.7\linewidth]{figures_paper/two_players_b_50-legend.pdf}}\vspace{-2em}\\
	\subcaptionbox{$R=1$}{\includegraphics[width=.3\linewidth]{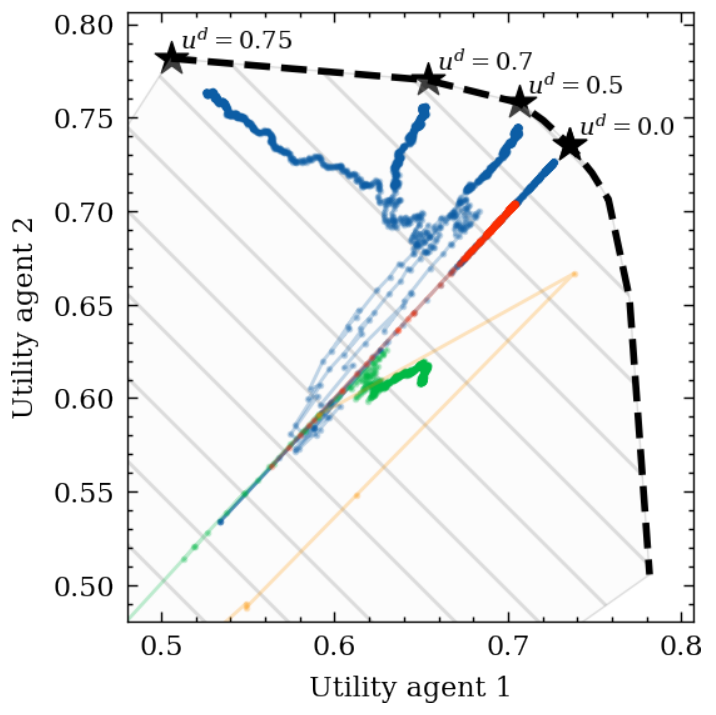}}
	\subcaptionbox{$R=50$}{\includegraphics[width=.3\linewidth]{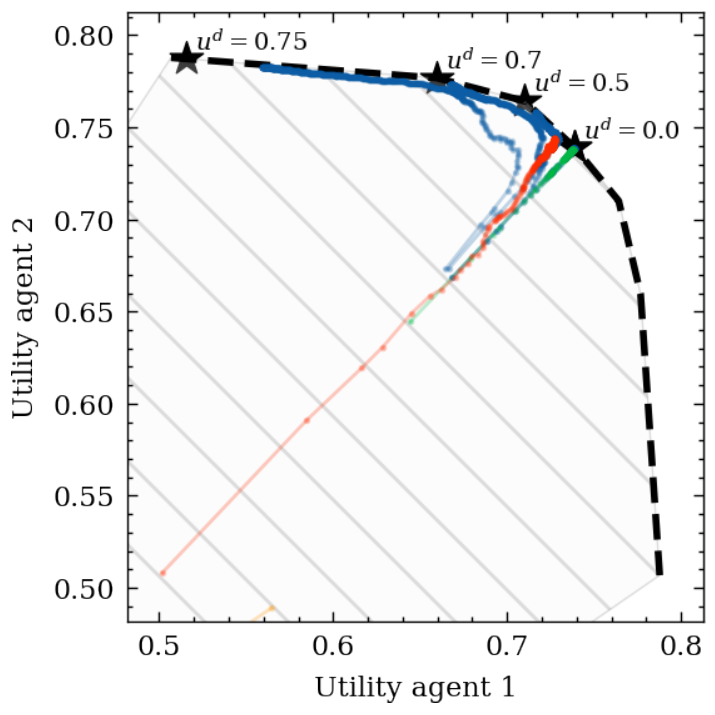}}
	\caption{Time-averaged utilities obtained for policies \horizonalgoname, \slotalgoname, \lru, and \lfu{} for batch sizes (a)~$R = 1$ and (b)~$R=100$, under \SC{} network topology.
	 Markers correspond to  iterations in~$\set{100, 200, \dots, 10^4}$. 
	}
	\label{fig:nash_bargainig}
\end{figure}

	        
	  


\begin{figure}[t]
	\centering
		\subcaptionbox*{}{\includegraphics[trim={0 7.6cm 0 0},clip,width=0.5\linewidth]{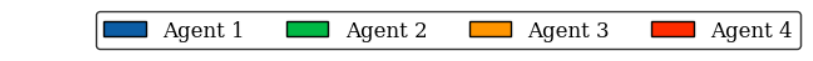}}\vspace{-2em}\\
	\subcaptionbox{$\alpha=1$\label{sfig:multiple_player11}}{\includegraphics[width=.3\linewidth]{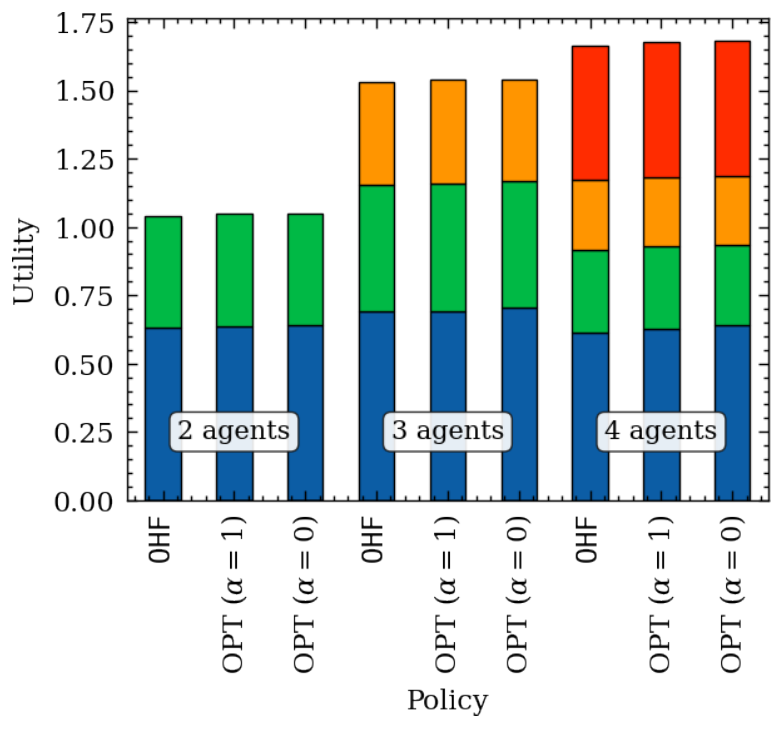}}
	\subcaptionbox{$\alpha=2$\label{sfig:multiple_player12}}{\includegraphics[width=.3\linewidth]{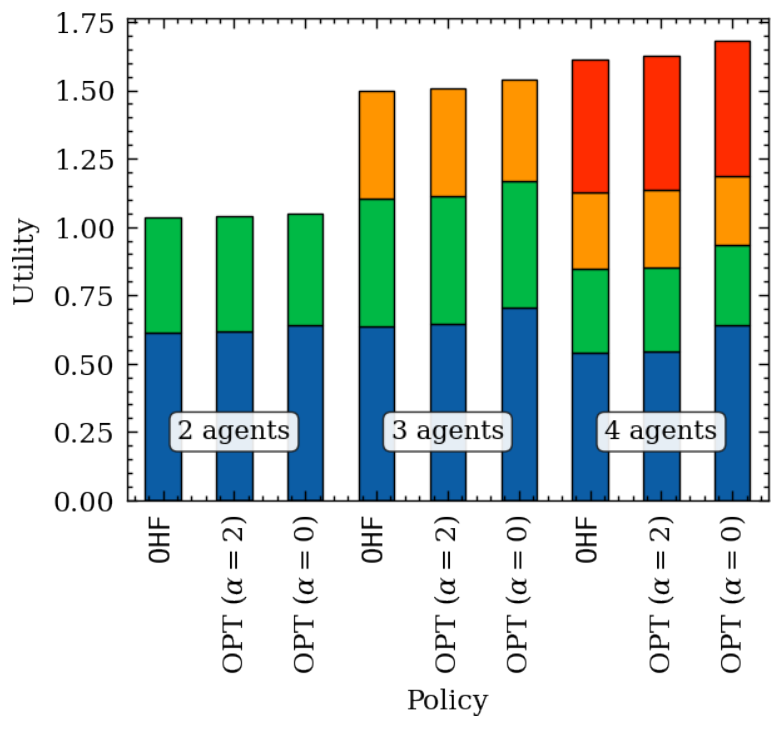}}
	\subcaptionbox{$\alpha=3$\label{sfig:multiple_player13}}{\includegraphics[width=.3\linewidth]{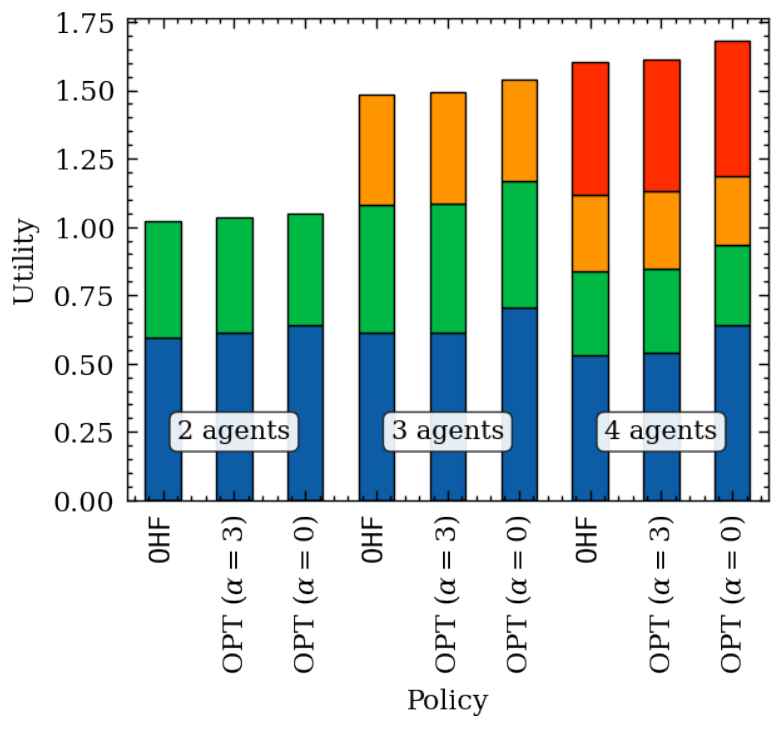}}
	
	\subcaptionbox{PoF\label{sfig:multiple_player14}}{\includegraphics[width=.3\linewidth]{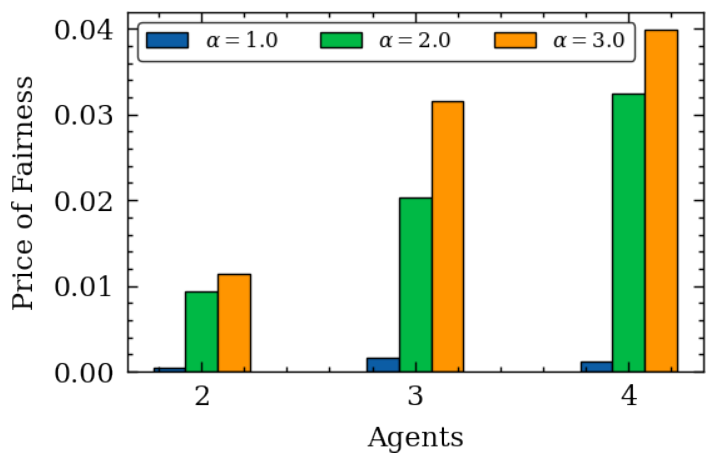}}
	\caption{Subfigures~(a)--(c) provide the average utility for different agents obtained by \horizonalgoname{}, fairness benchmark~(OPT for $\alpha \neq 0$), and  utilitarian benchmark (OPT for $\alpha =0$); and Subfigure~(d) provides the PoF for 
	$\alpha \in \set{0, 1, 2, 3}$ under an increasing number of agents in~$\set{2,3,4}$ and \BT~1--3 network topology.}
	\label{fig:multiple_player1}
	
\end{figure}
\begin{figure}[t]
	\centering
	\subcaptionbox*{}{\includegraphics[trim={0 7.4cm 0 0},clip,width=0.7\linewidth]{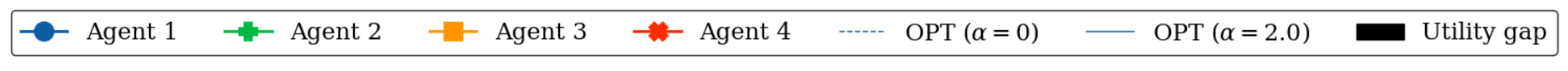}}\vspace{-2em}\\
	\subcaptionbox{}{\includegraphics[width=0.3\linewidth]{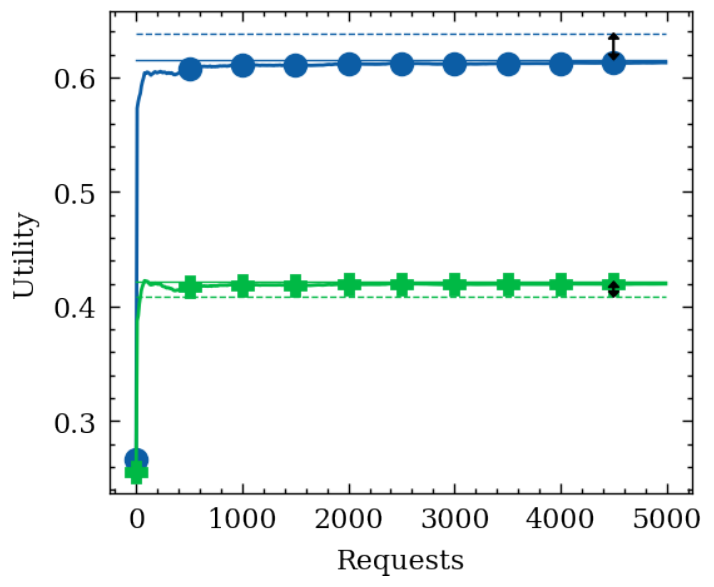}}
	\subcaptionbox{}{\includegraphics[width=0.3\linewidth]{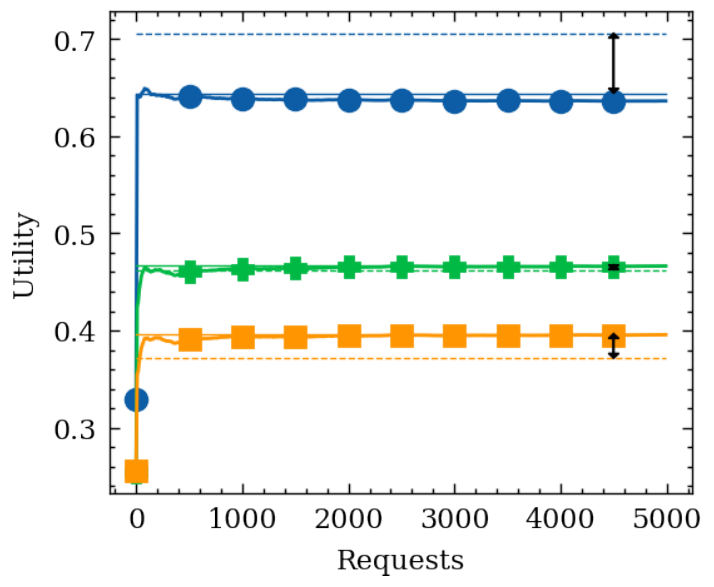}}  \subcaptionbox{}{\includegraphics[width=0.3\linewidth]{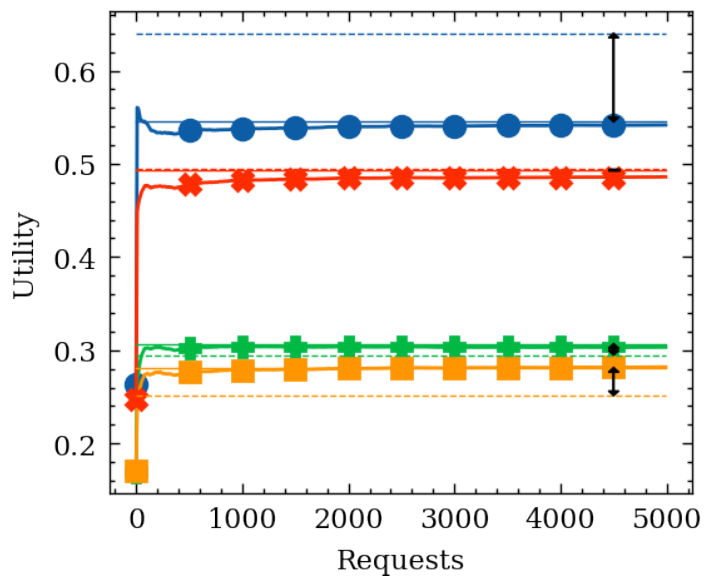}} 
	\caption{Subfigures (a)--(c) provide the time-averaged utility across different agents obtained by \horizonalgoname{} policy  and \OPT{} for $\alpha = 2$ under an increasing number of agents in $\set{2,3,4}$ and \BT~1--3 network topology.}
	\label{fig:multiple_player2}
	
\end{figure}

\begin{figure}
    \centering
	\subcaptionbox{\label{sfig:multiple_topologies}}{\includegraphics[width=0.3\linewidth]{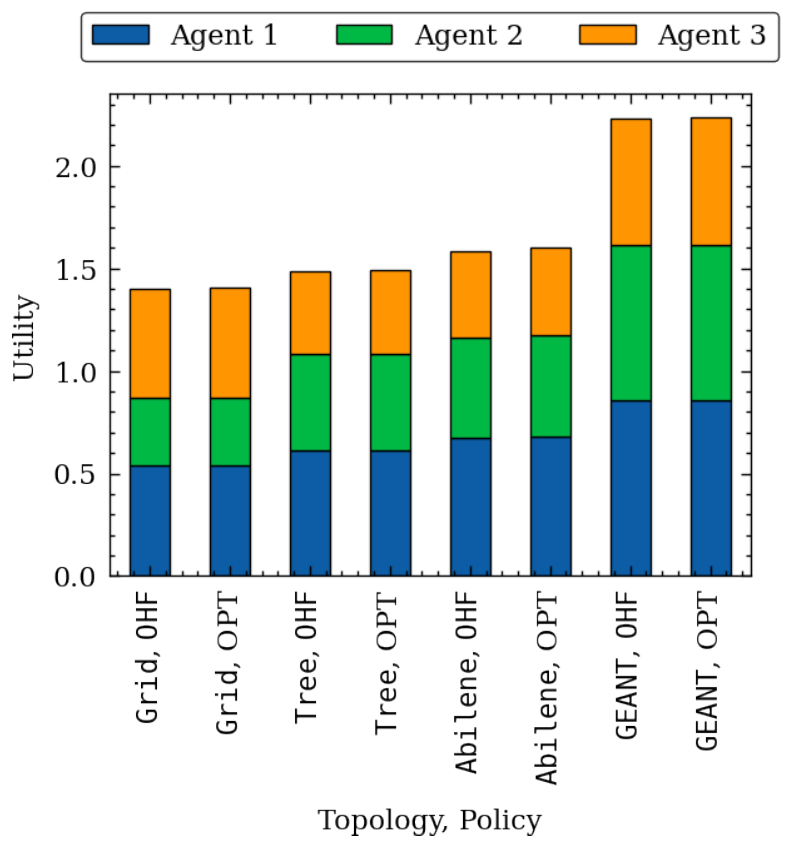}}
    \subcaptionbox{\label{sfig:adversarial_trace}}{\includegraphics[width =0.35\linewidth]{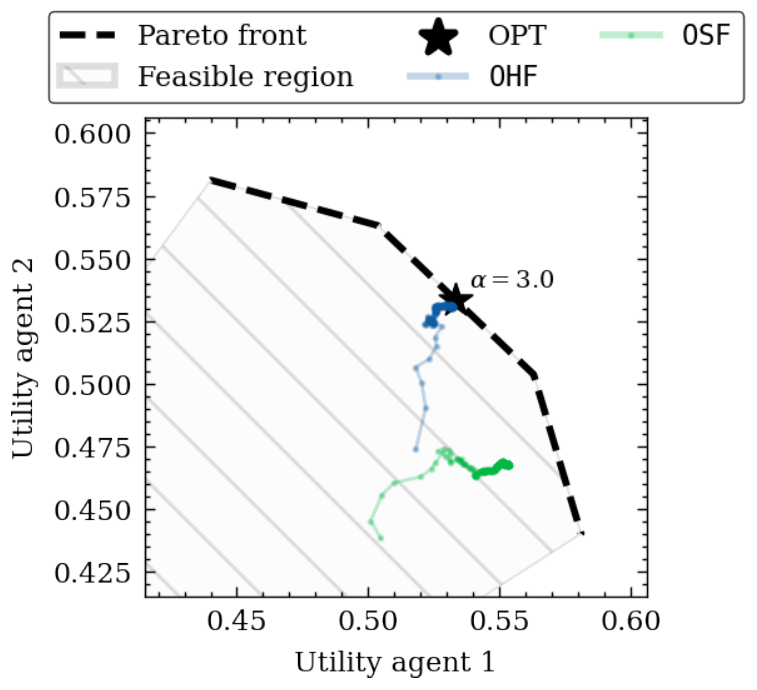}}  
    
    \caption{Subfigure~(a) provides the average utility of \horizonalgoname{} and fairness benchmark  under network topologies \BT, \Grid, \Abilene, \GEANT, \emph{Stationary} trace ($\sigma \in \set{0.6, 0.8, 1.2}$), and $\alpha =3$. Subfigure~(b) provides the time-averaged utilities obtained for \horizonalgoname, \slotalgoname{} and batch size $R = 50, t \in \T$, under network topology Tree~(a) and \emph{Non-stationary} trace. The markers represent the iterations in the set $\set{100, 200, \dots, 10^4}$.}
    \label{fig:multiple_topologies_adversarial_trace}
\end{figure}

\paragraph{Impact of  agents on the price of fairness.} In Figures~\ref{fig:multiple_player1} and~\ref{fig:multiple_player2}, we consider the \BT{}~1--3 topology,  $\alpha \in \set{1, 2, 3}$, and $|\mathcal I | \in \set{2, 3, 4}$. Agents' query nodes generate \emph{Stationary} trace~($\sigma \in \set{1.2,0.8,0.6}$).


In Figures~\ref{fig:multiple_player1}~(a)--(c), we observe for increasing the number of agents, the division of  utilities differs between the fairness benchmark and utilitarian benchmark; moreover, this difference is more evident for larger values of $\alpha$. Figure~\ref{fig:multiple_player1}~(d) provides the price of fairness, and we observe the price of fairness increases with the number of agents and $\alpha$. Nonetheless, under the different settings the price of fairness remains below $4\%$, i.e., we experience at most a $4\%$ drop in the social welfare to provide fair utility distribution across the different agents. Figure~\ref{fig:multiple_player2} gives the time-averaged utilities obtained by running \horizonalgoname{} for $\alpha =2$. We observe the utilities obtained by \horizonalgoname{} quickly converge to the same utilities obtained by the fairness benchmark. In this figure, we also highlight the difference between the utilities achieved by the fairness benchmark and utilitarian benchmark, is reflected by the increasing utility gap for a higher number of participating agents.

\paragraph{Different network topologies.} In Figure~\ref{fig:multiple_topologies_adversarial_trace}~(a), we consider the network topologies \BT, \Grid, \Abilene, \GEANT{} under \emph{Stationary} trace ($\sigma \in \set{0.6, 1.0, 1.2}$) and $\alpha =3$. \horizonalgoname{}  achieves the same utilities as the fairness benchmark across the different topologies. Note that for larger network topologies agents achieve a higher utility because there are more resources available.
\paragraph{Impact of non-stationarity.} In Figure~\ref{fig:multiple_topologies_adversarial_trace}~(b), we consider the \SC{} topology and $\alpha = 3$. The query node of agent~1 generates \emph{Non-Stationary} trace, while the query node of agent~2 generates a shuffled  \emph{Non-Stationary} trace, i.e., we remove the non-stationarity from the trace for agent~2 while preserving the overall popularity of the requests. Therefore, on average the agents are symmetric and experience the same utilities. We observe in Figure~\ref{fig:multiple_topologies_adversarial_trace}~(b) that indeed this is the case for \horizonalgoname{} policy; however, because \slotalgoname{} aims to insure fairness across the different timeslots the agents are not considered symmetric and the average utilities deviate from the Pareto front (not efficient). \slotalgoname{} favors agent~1 by increasing  its utility by $20\%$ compared to the utility of agent~1.

\section{Conclusion and Future Work}
\label{s:conclusion}

In this work, we proposed a novel  \algoname{} policy that achieves horizon-fairness in dynamic resource allocation problems. 
 We demonstrated the applicability  of this policy in virtualized caching systems where different agents can cooperate to increase their caching gain.  Our work paves the road for several interesting next steps. A future research direction is to consider decentralized versions of the policy under which each agent selects an allocation with limited information exchange across agents. For the application to virtualized caching systems, the message exchange techniques in~\cite{Ioannidis16,Li2021} to estimate subgradients can be exploited. 
Another important future research direction is to bridge the horizon-fairness and slot-fairness criteria to target applications where the agents are interested in ensuring fairness within a target time window. We observed that \algoname{} can encapsulate the two criteria, however,  it remains an open question whether a policy can smoothly transition between them. A final interesting research direction is to consider a limited feedback scenario where only part of the utility is revealed to the agents \new{(e.g., bandit feedback)}. Our policy could be extended to this setting through gradient estimation techniques~\cite{Hazanoco2016}.

\section{Acknowledgments}
This publication has emanated from research conducted with the financial support of the European Commission through Grant No. 101017109 (DAEMON). This research was supported in part by the French Government through the ``Plan de Relance'' and ``Programme d’investissements d’avenir''.

\clearpage

\bibliographystyle{ACM-Reference-Format}
\bibliography{Bibliography,Related,george}
\clearpage
\appendix
\section{Technical Lemmas and Definitions}
\subsection{Convex Conjugate}
\begin{definition}
\label{def:conjugate}
Let $F: \U \subset \reals^\I \to \reals \cup \{-\infty,+\infty\}$ be a function. Define ${F}^\star: \reals^\I \to \reals \cup \{-\infty,+\infty\}$  by
\begin{align}
   {F}^\star (\vec \theta ) = \sup_{\u \in \U} \left\{\u \cdot \dv - F(\u)\right\},
\end{align}
for $\dv \in \reals^\I$. This is  the  \emph{convex conjugate}  of $F$.
\end{definition}
\subsection{Convex Conjugate  of $\alpha$-Fairness Function}
\begin{lemma}
\label{lemma:convex_conjugate}
Let $\U = \brackets{u_{\star, \min}, u_{\star, \max}}^\I \subset \reals^\I_{> 0}$, $\Theta = \brackets{-1/u^\alpha_{\star, \min},- 1/u^\alpha_{\star, \max}}^\I \subset \reals^\I_{<0}$, and $F_\alpha: \U  \to \reals$ be an $\alpha$-fairness function~\eqref{e:alpha-fair}. The convex conjugate of $-F_\alpha$ is given by
\begin{align}
    \parentheses{-F_{\alpha}}^\star(\vec \theta) &= \begin{cases}\displaystyle
    \sum_{i \in \I} \frac{\alpha(-\theta_i)^{1-1/\alpha} - 1}{1-\alpha} &  \text{ for $\alpha \in \reals_{\geq 0} \setminus \set{1}$}, \\ 
     \displaystyle \sum_{i \in \I}  - \log(-\theta_i) - 1   &\text{ for $\alpha  = 1$},
    \end{cases}
\end{align}
where $\dv \in  \Theta$.
\end{lemma}
\begin{proof}
The convex conjugate of $-f_{\alpha}(u) \triangleq -\frac{u^{1-\alpha} - 1}{1-\alpha}$ for $u \in \brackets{u_{\star,\min}, u_{\star,\max}}$ and $\alpha \in \reals_{\geq 0}\setminus \set{1}$ is given by \begin{align}
\parentheses{-f_{\alpha}}^\star(\theta) = \max_{u \in \brackets{u_{\star,\min}, u_{\star,\max}}} \set{u \theta +\frac{u^{1-\alpha} - 1}{1-\alpha}}.\label{e:sup-fenchel}
\end{align}
We evaluate the derivative to characterize the maxima of r.h.s. term in the above equation
\begin{align}
    \frac{\partial}{\partial u} \parentheses{u \theta +\frac{u^{1-\alpha} - 1}{1-\alpha}} = \theta + \frac{1}{u^\alpha}.
\end{align}
The function $\theta + \frac{1}{u^\alpha}$ is a decreasing function in $u$; thus $\theta + \frac{1}{u^\alpha} \geq 0$ when  $u \leq \parentheses{-\frac{1}{\theta}}^{\frac{1}{\alpha}}$, and $\theta + \frac{1}{u^\alpha} <  0$ otherwise. The maximum is achieved at $u = \parentheses{-\frac{1}{\theta}}^{\frac{1}{\alpha}}$. It holds through Eq.~\eqref{e:sup-fenchel}
\begin{align}
    \parentheses{-f_{\alpha}}^\star(\theta) &= \frac{\alpha(-\theta)^{1-1/\alpha} - 1}{1-\alpha} &\text{for}\qquad \theta \in \brackets{-1/u^\alpha_{\star, \min},-1/u^\alpha_{\star, \max}}.\label{e:fenchel-conjugate1}
\end{align}
Moreover, it is easy to check that the same argument holds for $f_{1} (u) = \log(u)$ and we have
\begin{align}
    \parentheses{-f_{1}}^\star(\theta) &= -1 - \log(-\theta)&\text{for}\qquad \theta \in \brackets{-{1}/{u_{\star, \min}},-{1}/{u_{\star, \max}}}.\label{e:fenchel-conjugate2}
\end{align}
The convex conjugate of $-F_\alpha(\u) = \sum_{i \in \I} f_{\alpha} (u_i)$ for $\u \in \U $, using Eq.~\eqref{e:fenchel-conjugate1} and Eq.~\eqref{e:fenchel-conjugate2}, is given by
\begin{align}
    \parentheses{-F_{\alpha}}^\star(\vec \theta) =  \sum_{i \in \I} \parentheses{-f_{\alpha}}^\star( \theta_i)  &= \begin{cases}\displaystyle
    \sum_{i \in \I} \frac{\alpha(-\theta_i)^{1-1/\alpha} - 1}{1-\alpha} &  \text{ for $\alpha \in \reals_{\geq 0} \setminus \set{1}$}, \\ 
     \displaystyle \sum_{i \in \I}  - \log(-\theta_i) - 1   &\text{ for $\alpha  = 1$},
    \end{cases}
\end{align}
for $\dv \in \Theta$, because $F_\alpha(\u)$ is separable in $\u \in \U$. 

\end{proof}

\subsection{Convex Biconjugate of $\alpha$-Fairness Functions}
The following Lemma provides  a stronger condition on $\vec \theta$ compared to \cite[Lemma 2.2]{agrawal2014bandits}, i.e., we restrict  $\dv \in \Theta$ instead of $\norm{\dv}_{\star} \leq L$ where $L \geq \norm{ \nabla_{\u} F_\alpha(\u)}_{\star}$ for all $\u \in \U$ and $\norm{\spacedcdot}_\star$ is the dual norm of $\norm{\spacedcdot}$.
\begin{lemma}
\label{l:recover_f}
Let $\U = \brackets{u_{\star,\min}, u_{\star,\max}}^\I \subset \reals^\I_{> 0}$,  $\Theta = \brackets{-1/u_{\star, \min}^\alpha,- 1/u_{\star, \max}^\alpha}^\I\subset \reals^\I_{<0}$, and $F_\alpha: \U \to \reals$ be an $\alpha$-fairness function~\eqref{e:alpha-fair}. The function $F_{\alpha}$ can be always be recovered from the convex conjugate $\parentheses{-F_\alpha}^{\star}$, i.e.,
\begin{align}
    F_{\alpha}(\u) = \min_{\dv \in \Theta} \left\{  \parentheses{-F_{\alpha}}^\star(\vec \theta) - \vec \theta \cdot \u\right\},
\end{align}
for $\u \in \U$.
\end{lemma}
\begin{proof}
This proof follows the same lines of the proof of~\cite[Lemma 2.2]{agrawal2014bandits}. Since $\vec u \in \U$, therefore the gradient of $F_\alpha$ at point $\vec u$ is given as  $\nabla_{\u} F_\alpha(\vec u) = \brackets{1/{u^{\alpha}_i}}_{i \in \I} \in -\Theta =\brackets{1/u_{\star, \min}^\alpha, 1/u_{\star, \max}^\alpha}^\I$. Moreover, it holds 
\begin{align}
    \min_{\dv \in \Theta} \left\{ \parentheses{-F_\alpha}^\star(\vec\theta) - \vec \theta \cdot \u\right\} &= \min_{\dv \in \Theta} \left\{ \max_{\vec u' \in \U}\left\{ \vec \theta\cdot \u' + F_\alpha(\u) \right\}- \vec \theta \cdot \u\right\}\\
    &= \max_{\vec u' \in \U} \min_{\dv \in \Theta} \left\{  \vec \theta\cdot \u' + F_\alpha(\u') - \vec \theta \cdot \u\right\}. &\text{Minmax theorem}
    \label{e:minmax_eq}
\end{align}
We take 
\begin{align*}
    \min_{\dv \in \Theta} \left\{  \vec \theta\cdot \u' + F_\alpha(\u') - \vec \theta \cdot \u\right\} &=    \min_{\dv \in \Theta} \left\{F_\alpha(\u') + \vec \theta \cdot  (\u' - \u)\right\}\\
    &\leq F_\alpha(\u') -\nabla F_\alpha(\u) \cdot  (\u' - \u) &\text{Because $-\nabla F_\alpha(\u) \in  \Theta$}\\ 
    &\leq F_\alpha(\vec u). &\text{Use concavity of $F_\alpha$}
\end{align*}
The equality is achieved when $\u = \u'$ and the maximum value in \eqref{e:minmax_eq} is attained for this value. We conclude the proof.
\end{proof}
\subsection{Online Gradient Descent (OGD) with Self-Confident Learning Rates}
Lemma~\ref{lemma:ogd_regret} provides the regret guarantee of OGD oblivious to the time horizon $T$ and bound on subgradients' norm for any $t \in \T$. This adopts the idea of~\cite{AUER200248} which  denominate such learning schemes as self-confident. 
\begin{lemma}
\label{lemma:ogd_regret}
Consider a convex set $\X$, a sequence of $\sigma$-strongly convex functions $f_t: \X \to \reals$ with subgradient $\vec g_t \in \partial f_t(\x_t)$ at $\x_t$, and OGD update rule $\x_{t+1} = \Pi_{\X} \parentheses{\x_{t} -  \eta_t \vec g_t } = \arg\min_{\x \in \X} \norm{\x - \parentheses{\x_{t} - \eta_t\vec g_t}}_2$  initialized at $\x_1 \in \X$.  Let $\diam {\X} \triangleq \max\set{\norm{\x - \x'}_2: \x, \x' \in \X }$. Selecting the learning rates as $\vec \eta: \T \to \reals$ such that $\eta_{t} \leq  \eta_{t-1}$ for all $t > 1$ gives the following regret guarantee against a fixed decision $\x \in \X$:
\begin{align}
    \sum_{t \in \T} f_t(\x_t)-f_t(\x) \leq \mathrm{diam}^2(\X) \sum^T_{t=1} \parentheses{\frac{1}{\eta_t} - \frac{1}{\eta_{t-1}} - \sigma} +  \sum^T_{t=1}\eta_t \norm{\vec g_t}^2_2.
\end{align}
\begin{itemize}
    \item When $\sigma >0$, selecting the learning rate schedule $\eta_t = \frac{1}{\sigma t}$ for $t \in \T$ gives 
\begin{align}
    \sum_{t \in \T} f_t(\x_t)-f_t(\x) &\leq \sum^T_{t=1}\frac{\norm{\vec g_t}^2_2}{t\sigma}  = \BigO{\log(T)}.
\end{align}
\item When $\sigma = 0$, selecting the learning rate schedule $\eta_t = \frac{\diam{\X}}{\sqrt{ \sum^t_{s=1} \norm{\vec g_s}^2_2}}$ for $t \in \T$  gives
\begin{align}
        \sum_{t \in \T} f_t(\x_t)-f_t(\x) &\leq 1.5 \diam{\X} \sqrt{ \sum_{t\in \T} \norm{\vec g_s}^2_2} = \BigO{\sqrt{T}}.
\end{align}
\end{itemize}
\end{lemma}
\begin{proof}
This proof follows the same lines of the proof of~\cite{Hazanoco2016}. We do not assume  a bound on the gradients is known beforehand and the time horizon $T$.
Take a fixed $\x \in \X$. Applying the definition of $\sigma$-strong convexity to the pair of points $\x_t$ and $\x$, we have
\begin{align}
    2 \parentheses{f_t(\x_t) - f_t(\x) }\leq 2 \vec g_t \cdot (\x_t - \x) -\sigma \norm{\x_t - \x}^2_2.\label{eq:ogd1}
\end{align}
Pythagorean theorem implies 
\begin{align}
    \norm{\x_{t+1} - \x}_2^2 = \norm{\Pi_{\X}\parentheses{\x_t - \eta_t \vec g_t} - \x}_2^2 \leq \norm{\x_t - \eta_t - \x}_2^2,
\end{align}
Expanding the r.h.s. term gives
\begin{align}
    \norm{\x_{t+1} - \x}_2^2 &\leq \norm{\x_t - \x}_2^2 + \eta^2_t \norm{\vec g_t}^2_2 - 2\eta_t \vec g_t \cdot \parentheses{\x_t - \x}, \\
   2\vec g_t \cdot \parentheses{\x_t - \x}&\leq\frac{ \norm{\x_t - \x}_2^2 -  \norm{\x_{t+1} - \x}_2^2}{\eta_t} + \eta_t \norm{\vec g_t}^2_2.\label{eq:ogd2}
\end{align}
Combine Eq.~\eqref{eq:ogd1} and Eq.~\eqref{eq:ogd2} and for $t =1 $ to $t = T$:
\begin{align*}
    2\sum^T_{t =1} f_t(\x_t) - f_t(\x) &\leq \sum^T_{t=1} \frac{ \norm{\x_t - \x}_2^2 (1-\sigma \eta_t) -  \norm{\x_{t+1} - \x}_2^2}{\eta_t} + \sum^T_{t=1}\eta_t \norm{\vec g_t}^2_2 \\
    &\leq \sum^T_{t=1} \norm{\x_t - \x}_2^2 \parentheses{\frac{1}{\eta_t} - \frac{1}{\eta_{t-1}} - \sigma} +  \sum^T_{t=1}\eta_t \norm{\vec g_t}^2_2 &\text{$\frac{1}{\eta_0} \triangleq0$}\\
    &\leq   \mathrm{diam}^2(\X)  \parentheses{\frac{1}{\eta_T} -\sigma T } +  \sum^T_{t=1}\eta_t \norm{\vec g_t}^2_2.&\text{Telescoping series}
    \label{eq:generic_bound}
\end{align*}

\noindent When $\sigma > 0$ and $\eta_t = \frac{1}{\sigma t}$, from Eq.~\eqref{eq:generic_bound} we have
\begin{align}
    \sum^T_{t =1} f_t(\x_t) - f_t(\x) \leq 0 + \sum^T_{t=1} \frac{\norm{\vec g_t}^2_2}{2\sigma t} \leq \max_{t \in \T} \set{\norm{\vec g_t}^2_2} \sum^T_{t =1} \frac{1}{2\sigma} \leq  \frac{\max_{t \in \T} \set{\norm{\vec g_t}^2_2}}{2\sigma} \mathrm{H}_T = \BigO{\log(T)},
\end{align}
where $\mathrm{H}_T$ is the $T$-th harmonic number.

\noindent When $\sigma = 0$ and $\eta_t = \frac{\diam{\X}}{\sqrt{ \sum^t_{s=1} \norm{\vec g_s}^2_2}}$, from Eq.~\eqref{eq:generic_bound} we have
\begin{align}
    \sum^T_{t =1} f_t(\x_t) - f_t(\x) &\leq  \frac{\diam{\X}}{2}  {\sqrt{ \sum^T_{t=1} \norm{\vec g_s}^2_2}} +  \frac{ \diam{\X} }{2} \sum^T_{t=1} \frac{\norm{\vec g_t}^2_2}{\sqrt{ \sum^t_{s=1} \norm{\vec g_s}^2_2}}\\
    &\leq 1.5  \diam{\X}{\sqrt{ \sum^T_{t=1} \norm{\vec g_s}^2_2}} = \BigO{\sqrt{T}}.
\end{align}
Last inequality is obtained using \cite[Lemma~3.5]{AUER200248}, i.e., $\textstyle{\sum^T_{t=1} \frac{\abs{a_t}}{\sum^t_{s=1} \abs{a_s}}} \leq 2 \sqrt{\sum^T_{t=1} \abs{a_t}}$. This  concludes the proof. 
\end{proof}

\subsection{Saddle-Point Problem Formulation of $\alpha$-Fairness}
\begin{lemma}
\label{lemma:properties_saddle_function}
\label{l:saddle_problem}
Let $\X$ be a convex set,  $\U = \brackets{u_{\star,\min}, u_{\star,\max}}^\I \subset \reals^\I_{> 0}$, $u_i: \X \to \U$ be a  concave function for every $i \in \I$, $\Theta = \brackets{-1/u_{\star, \min}^\alpha,- 1/u_{\star, \max}^\alpha}^\I\subset \reals^\I_{<0}$, and  $\Psi_\alpha: \Theta \times \X \to \reals $ be a function given by
\begin{align}
    \Psi_{\alpha} (\dv, \x) \triangleq \parentheses{-F_\alpha}^\star(\vec \theta) - \vec \theta \cdot \vec u(\x).
\end{align}
The following holds:
\begin{itemize}
    \item The solution of the saddle-point problem formed by $\Psi_{\alpha}$ is a maximizer of the $\alpha$-fairness function
    \begin{align}
        \max_{\x \in \X} \min_{\dv \in \Theta}    \Psi_{\alpha} (\dv, \x) = \max_{\x \in \X} F_{\alpha} (\u(\x)).\label{e:sp-1}
    \end{align}
    \item The function $\Psi_{\alpha}: \Theta \times \X \to \reals $ is concave over $\X$.
    \item The function $\Psi_{\alpha}: \Theta \times \X \to \reals $ is $ \frac{u_{\star, \min}^{1+1/\alpha}}{\alpha }$-strongly convex over $\Theta$ w.r.t. $\norm{\spacedcdot}_2$ for $\alpha > 0$.
\end{itemize}
\end{lemma}
\begin{proof}
Equation~\eqref{e:sp-1} is a direct result of Lemma~\ref{l:recover_f}. The function  $\Psi_{\alpha}$ is concave over $\X$ because $ - \vec \theta \cdot \vec u(\x)$ is a weighted sum of concave functions with non-negative weights. To prove the strong convexity of $\Psi_{\alpha}$ w.r.t. $\norm{\spacedcdot}_2$, a sufficient condition~\cite[Lemma~14]{Shalev2007} is given by ${\dv'}^T \parentheses{\nabla_{{\dv}}^2 \Psi_{\alpha} ({\dv}, \x)} {\dv}' \geq \sigma \norm{ {\dv}'}^2_2$ for all $\dv, {\dv}' \in \Theta$, and it holds
\begin{align}
    {\dv'}^T \parentheses{\nabla_{{\dv}}^2 \Psi_{\alpha} ({\dv}, \x)}  {\dv}' = \sum_{i \in \I} {\theta_i'}^2 \frac{\partial^2}{\partial {\theta_i}} \parentheses{-F_{\alpha}}^\star({\dv}) = \sum_{i \in \I}  \frac{{\theta'_i}^2}{\alpha (-\theta_i)^{1+1/\alpha}} \geq \frac{u_{\star, \min}^{1+1/\alpha}}{\alpha } \norm{\vec{\theta}'}^2_2.
\end{align}
This concludes the proof.
\end{proof}

\section{Proof of Theorem~\ref{theorem:impossibility}}
\label{proof:impossibility}
\begin{proof}
Consider two players $\I = \set{1,2}$, allocation set $\mathcal{X} = [-1,1]$ for all $t \in \T$. We define  $\gamma_T \in [0.4,1]$,  $\psi_T\triangleq\frac{1}{T}\sum^{\gamma_T T}_{t=1} {x_t}$. We assume w.l.g. $\gamma_T T$ is a natural number. We consider two strategies selected by the adversary:

\noindent\textbf{Strategy 1.} The adversary reveals the following utilities:
\begin{align}
    \vec u_t (x) &= \begin{cases} (1 + x, 2 - x) & \text{if } t \leq \gamma_T T, \\
    (1, 1) & \text{otherwise}.
    \end{cases}
\end{align}
Under the selected utilities, the static optimum attains the following objective 
\begin{align}
    \mathrm{OPT}^{\mathrm{S1}} &= \max_{x \in \mathcal{X}}~ f_{\alpha}((1+x) \gamma_T +  (1-\gamma_T))  +f_{\alpha}((2-x) \gamma_T + (1-\gamma_T))\\
    &=\max_{x \in \mathcal{X}}~ f_{\alpha}(1 +\gamma_T x)  +f_{\alpha}( 1 +\gamma_T - \gamma_T x).
\end{align}
The above objective is concave in $x$. We can perform a derivative test to characterize its maximum
\begin{align}
    \frac{\partial f_{\alpha}(1 +\gamma_T x)  +F_{\alpha}( 1 +\gamma_T - \gamma_T x)}{\partial x} = \frac{\gamma_T}{ \parentheses{1 + \gamma_T x}^{\alpha}} - \frac{\gamma_T}{\parentheses{1 + \gamma_T - \gamma_T x}^\alpha} = 0,\qquad \text{for}~x = \frac{1}{2}.
\end{align}
Thus, it holds
\begin{align}
     \mathrm{OPT}^{\mathrm{S1}} = 2 f_{\alpha}( 1 + 0.5\gamma_T).
\end{align}
The fairness regret denoted by  $\mathfrak{R}_T^{\mathrm{S1}} (F_\alpha, \vec A)$ under this strategy of a policy $\mathcal{A}$ is given by
\begin{align}
\mathfrak{R}_T^{\mathrm{S1}} (F_\alpha, \vec A) &= \mathrm{OPT}^{\mathrm{S1}} - f_{\alpha}\left(\frac{1}{T}\left(\sum^{\gamma_T T}_{t=1} {1 + x_t} \right)  +  1 - \gamma_T\right) - f_{\alpha}\left(\frac{1}{T}\left(\sum^{\gamma_T T}_{t=1} {2 - x_t} \right) +1 -\gamma_T\right)\\
&= 2 f_{\alpha}( 1 + 0.5\gamma_T) - f_{\alpha}(1 +\psi_T) - f_{\alpha}(1 + \gamma_T - \psi_T).
\end{align}

\noindent\textbf{Strategy 2.} The adversary reveals the following utilities:
\begin{align}
    \vec u_t (x) &= \begin{cases} (1 + x, 2 - x) & \text{if } t \leq \gamma_T T, \\
    (2, 0) & \text{otherwise}.
    \end{cases}
\end{align}
Under the selected utilities, the static optimum attains the following objective 
\begin{align}
    \mathrm{OPT}^{\mathrm{S2}} &= \max_{x \in \mathcal{X}} f_{\alpha} ((1+x)\gamma_T + (1-\gamma_T)2) +  f_{\alpha} ((2-x)\gamma_T)\\
    &=\max_{x \in \mathcal{X}} f_{\alpha} (2 - \gamma_T + \gamma_T x) +  f_{\alpha} (2\gamma_T - \gamma_T x ).
\end{align}
Similar to the previous strategy, we can perform a derivative test to characterize the maximum of the the above objective
\begin{align*}
    \frac{\partial f_{\alpha} (2 - \gamma_T + \gamma_T x) +  f_{\alpha} (2\gamma_T - \gamma_T x )}{\partial x} = \frac{\gamma_T}{\parentheses{2 - \gamma_T + \gamma_T x}^\alpha} -\frac{\gamma_T}{\parentheses{2\gamma_T - \gamma_T x}^\alpha} = 0, \qquad \text{for}~ x = 1.5 - \frac{1}{\gamma_T}.
\end{align*}
Therefore, it holds
\begin{align}
     \mathrm{OPT}^{\mathrm{S2}} = 2f_{\alpha}( 1 + 0.5 \gamma_T).
\end{align}
and the fairness regret $\mathfrak{R}_T^{\mathrm{S2}} (F_\alpha, \vec A)$ under this strategy is 
\begin{align}
    \mathfrak{R}_T^{\mathrm{S2}} (F_\alpha, \vec A) &=  \mathrm{OPT}^{\mathrm{S2}} - f_{\alpha}\left(\frac{1}{T}\left(\sum^{\gamma_T T}_{t=1} {1 + x_t} \right)  +  2 -2 \gamma_T\right) - f_{\alpha}\left(\frac{1}{T}\left(\sum^{\gamma_T T}_{t=1} {2 - x_t} \right)\right)\\
    &= 2 f_{\alpha}( 1 + 0.5 \gamma_T)  - f_{\alpha}(2 - \gamma_T + \psi_T) - f_{\alpha}(2 \gamma_T - \psi_T).
\end{align}
We take the average fairness regret $\frac{1}{2}\left(  \mathfrak{R}_T^{\mathrm{S1}} (F_\alpha, \vec A) +  \mathfrak{R}_T^{\mathrm{S2}} (F_\alpha, \vec A)\right)$ across  the two strategies
\begin{align}
    &\frac{1}{2}\left(  \mathfrak{R}_T^{\mathrm{S1}} (F_\alpha, \vec A) +  \mathfrak{R}_T^{\mathrm{S2}} (F_\alpha, \vec A)\right)\\
    &= 2 f_{\alpha}(1 + 0.5 \gamma_T) - \frac{1}{2} \left(f_{\alpha}(2 - \gamma_T + \psi_T) + f_{\alpha}(2 \gamma_T - \psi_T) + f_{\alpha}(1 +\psi_T) + f_{\alpha}(1 + \gamma_T - \psi_T)\right).\label{e:two_regrets}
\end{align}
The r.h.s. of the above equation is convex in $\psi_T$, so its minimizer can be characterized through the derivative as follows
\begin{align}
  &\frac{\partial f_{\alpha}(2 - \gamma_T + \psi) + f_{\alpha}(2 \gamma_T - \psi) + f_{\alpha}(1 +\psi) + f_{\alpha}(1 + \gamma_T - \psi)}{\partial \psi} \\
  &= \frac{1}{\parentheses{2 - \gamma_T + \psi}^\alpha} - \frac{1}{\parentheses{2 \gamma_T - \psi}^\alpha} + \frac{1}{\parentheses{1 +\psi}^\alpha }- \frac{1}{\parentheses{1 + \gamma_T - \psi}^\alpha} = 0, \qquad \text{for } \psi = \gamma_T - 0.5.
\end{align}
We replace $\psi_T = \gamma_T - 0.5$ in Eq.~\eqref{e:two_regrets} to get
\begin{align}
    \frac{1}{2}\left(  \mathfrak{R}_T^{\mathrm{S1}} (F_\alpha, \vec A) +  \mathfrak{R}_T^{\mathrm{S2}} (F_\alpha, \vec A)\right) &\geq 2 f_{\alpha}(1 + 0.5 \gamma_T) - \left( f_{\alpha}(1.5) +  f_{\alpha}(0.5 + \gamma_T) \right).
    \label{e:two_regrets2}
\end{align}
We take the derivative of the lower bound 
\begin{align}
    \frac{\partial 2 f_{\alpha}(1 + 0.5 \gamma_T) -  \left(f_{\alpha}(1.5) + f_{\alpha}(0.5 + \gamma_T) \right)}{\partial \gamma_T}   = \frac{\parentheses{0.5 + \gamma}^\alpha - \parentheses{1 + 0.5 \gamma}^\alpha}{\parentheses{0.5 + \gamma}^\alpha \parentheses{1 + 0.5 \gamma}^\alpha}.
\end{align}
Note that the sign of the derivative is determined by the numerator  $\parentheses{0.5 + \gamma}^\alpha - \parentheses{1 + 0.5 \gamma}^\alpha$. It holds $\parentheses{0.5 + \gamma}^\alpha - \parentheses{1 + 0.5 \gamma}^\alpha < 0$ for $\gamma_T<1$, otherwise $\parentheses{0.5 + \gamma}^\alpha - \parentheses{1 + 0.5 \gamma}^\alpha = 0$. Hence, the lower bound in Eq.~\eqref{e:two_regrets2} is strictly decreasing for $\gamma_T < 1$, and it holds for $\gamma_T \leq 1 - \epsilon$ for $\epsilon >0$
\begin{align}
     \frac{1}{2}\left(  \mathfrak{R}_T^{\mathrm{S1}} (F_\alpha, \vec A) +  \mathfrak{R}_T^{\mathrm{S2}} (F_\alpha, \vec A)\right)&\geq 2 f_\alpha (1.5 - 0.5 \epsilon) - (f_\alpha(1.5) + f_\alpha (1.5 + 0.5 \epsilon)) > 0. 
\end{align}


In other words, the fairness regret guarantee is not attainable\footnote{Note that the fairness regret must vanish for any adversarial choice of sequence of utilities.} for values of $\gamma_T \leq 1 - \epsilon$ for any $T$ and $\epsilon > 0$. We can also verify that \ref{a:5} is violated when $\gamma_T \leq 1 - \epsilon$ for any $T$ and $\epsilon > 0$. Note that $\gamma_T$ is defined to be in the set $[0.4, 1]$.

Under strategy 1  we have $x_{\star} = \frac{1}{2}$ and it holds
\begin{align}
    \frac{1}{T} \sum^T_{t=1} \vec u_t(x_{\star})  &=  (1 + 0.5\gamma_T,1 + 0.5\gamma_T  ), \text{ and }
    \vec u_t (x_{\star}) = \begin{cases} (1.5, 1.5) & \text{if } t \leq \gamma_T T, \\
    (1, 1) & \text{otherwise}.
    \end{cases}
\end{align}

Then, it holds
\begin{align}
    \VT &\geq 2 ( 1 -\gamma_t) \gamma_T T  \geq 2 \epsilon \gamma_T T \geq 0.8 \epsilon T = \Omega(T).
\end{align}
Moreover, it can easily be checked that $\WT  = \Omega(T)$ because there is no decomposition $\set{1,2,\dots, T} =  \T_1\cup \T_2\cup\dots\cup\T_K$ where $\max\set{\T_k: k \in [K]} = o \left(T^{\frac{1}{2}}\right)$ under which $\sum^K_{k=1} \sum_{i \in \I} \abs{\sum_{t \in \T_k}  \delta_{t,i}(\x_{\star})}= o (T)$. 

To conclude, when $\gamma_T = 1 - o(1)$, we have $\min\{\VT, \WT\} \leq  \VT = o(T)$; thus, Assumption~\ref{a:5}  only holds when $\gamma_T = 1 - o(1)$ for which the vanishing fairness regret guarantee is attainable. Figure~\ref{fig:impossibility_example} provides a summary of the connection between the fairness regret under scenarios~1 and~2 and Assumption~\ref{a:5}.

\begin{figure}[t]
    \centering
    \includegraphics[width=.4\linewidth]{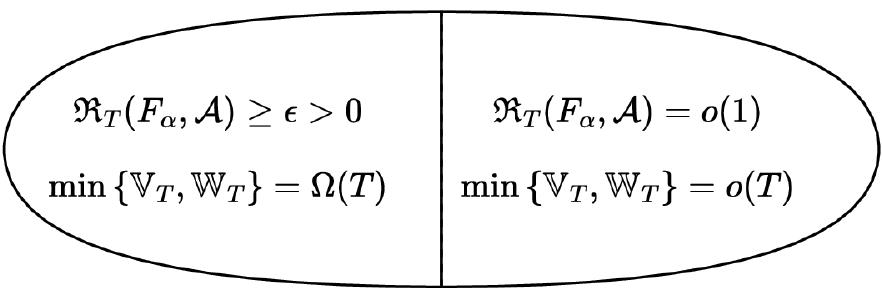}
    \caption{Assumption~\ref{a:5} and fairness regret~\eqref{e:b_regret} under scenarios~1 and 2.}
    \label{fig:impossibility_example}
\end{figure}

\end{proof}

\section{Proof of Theorem~\ref{th:maintheorem}}
\label{proof:t:maintheorem}

\begin{proof}
Note that $\Psi_{\alpha,t}: \Theta \times \X \to \reals$ is the function given by \begin{align}
     \Psi_{\alpha, t}(\dv,  \x) = \parentheses{-F_{\alpha}}^\star(\dv) - \dv \cdot \u_{t}(\x),
     \label{e:primal_dual_loss_gain}
\end{align}
where $F_\alpha: \U \to \reals$ is an $\alpha$-fairness function~\eqref{e:alpha-fair}.  From Lemma~\ref{lemma:ogd_regret}, OGD operating over the set $\Theta$ under the  $ \frac{u_{\star, \min}^{1+1/\alpha}}{\alpha }$-strongly convex (Lemma~\ref{lemma:properties_saddle_function})  cost functions $\Psi_{\alpha, t}(\underline{\vec \theta_t},  \x_t)$ 
has the following regret guarantee against any fixed $\dv \in \Theta$
\begin{align}
\label{e:algo_obj_p1}
    \frac{1}{T}\sum^T_{t=1} \Psi_{\alpha, t}(\dv_t, \x_t) -  \frac{1}{T}\sum^T_{t=1} \Psi_{\alpha, t}(\dv, \x_t) &\leq \frac{1}{T} \spacedcdot \vGroup{ \frac{1}{2}\sum^T_{t=1}\frac{\alpha}{u_{\star, \min}^{1+1/\alpha} t} \norm{\vec g_{\Theta, t}}^2_2}{\mathfrak{R}_{T, \Theta}} , 
\end{align}
From Lemma~\ref{l:recover_f}, it holds
\begin{align}
\min_{\dv \in \Theta} \frac{1}{T}\sum^T_{t=1} \Psi_{\alpha, t}(\dv, \x_t) &= F_{\alpha}\left(\frac{1}{T} \sum^T_{t=1}\vec u_t(\x_t)\right).\label{e:algo_obj_p2}
\end{align}
Combine Eq.~\eqref{e:algo_obj_p1} and Eq.~\eqref{e:algo_obj_p2} to obtain the lower bound
\begin{align}
    {F_{\alpha}\left(\frac{1}{T} \sum^T_{t=1}\u_t(\x_t)\right) }+  \frac{\mathfrak{R}_{T, \Theta}}{T} \geq { \frac{1}{T}\sum^T_{t=1} \Psi_{\alpha, t}(\dv_t, \x_t)}.
    \label{e:dual_part1}
\end{align}
From Lemma~\ref{lemma:ogd_regret}, OGD operating over the set $\X$ under the reward functions $\Psi_{\alpha, t}(\vec \theta_t,  \x)$
has the following regret guarantee for any fixed $\x_{\star} \in \X$:
\begin{align}
     \frac{1}{T}\sum^T_{t=1} \Psi_{\alpha, t}(\dv_t, \x_{\star}) - \frac{1}{T}\sum^T_{t=1} \Psi_{\alpha, t}(\dv_t, \x_t) 
    \leq \frac{1}{T} \spacedcdot \vGroup{1.5 \diam{\X} \sqrt{\sum_{t \in \T} \norm{\vec g_{\X,t}}^2_2}}{\mathfrak{R}_{T, \X}},
\end{align}
Hence, we have the following
\begin{align} 
 &{\frac{1}{T}\sum^T_{t=1} \Psi_{\alpha, t}(\dv_t, \x_t)} + \frac{\mathfrak{R}_{T, \X}}{T} \geq 
  {\frac{1}{T}\sum^T_{t=1} \Psi_{\alpha, t}(\dv_t, \x_{\star}) }\nonumber\\
   &= { \frac{1}{T}\sum^T_{t=1} F^\star(\dv_t) -\frac{1}{T} \sum^T_{t=1} \dv_t 
   \cdot \u_t(\x_{\star}) }&\text{Replace $\Psi_{\alpha, t}(\dv_t, \x_{\star})$ using Eq.~\eqref{e:primal_dual_loss_gain}}\nonumber\\
   &\geq  {F^\star\left(\frac{1}{T} \sum^T_{t=1}\dv_t\right) - \frac{1}{T} \sum^T_{t=1}\dv_t 
   \cdot \u_t(\x_{\star})} &\text{Jensen's inequality \& convexity of $F^\star$}\nonumber\\
   &\geq { F^\star\left(\bar{\dv}\right) - \bar{\dv} 
   \cdot \left(\frac{1}{T}  \sum^T_{t=1}\u_t(\x_{\star})\right)} -  {\frac{1}{T} \sum^T_{t=1} (\dv_t - \bar{\dv}) \cdot \u_t(\x_{\star})}\nonumber\\
   &\geq  \min_{\dv \in \Theta} \set{  {F^\star\left({\dv}\right) - {\dv} 
   \cdot \left(\frac{1}{T}  \sum^T_{t=1}\u_t(\x_{\star})\right) }}- { \frac{1}{T}  \sum^T_{t=1} (\dv_t - \bar{\dv}) \cdot \u_t(\x_{\star}) }\nonumber\\
   &= F_{\alpha}\left( {\frac{1}{T} \sum^T_{t=1} \u_t(\x_{\star})}\right) -  { \frac{1}{T}  \sum^T_{t=1}\left(\dv_t - \bar{\dv}\right) \cdot \u_t(\x_{\star})} .\label{e:primal_p2}
\end{align}
We combine the above equation and Eq.~\eqref{e:dual_part1} to obtain 
\begin{align}
{F_{\alpha}\left( {\frac{1}{T} \sum^T_{t=1}\u_t(\x_{\star})}\right) - F_{\alpha}\left( {\frac{1}{T} \sum^T_{t=1}\u_t(\x_t)}\right)} &\leq  {\frac{\mathfrak{R}_{T, \X}}{T} + \frac{\mathfrak{R}_{T, \Theta}}{T} + { \frac{1}{T} \sum^T_{t=1} \left(\dv_t - \bar\dv\right) \cdot \u_t(\x_{\star})}}\nonumber\\
&= { \frac{\mathfrak{R}_{T, \X}}{T} + \frac{\mathfrak{R}_{T, \Theta}}{T} } + \vGroup{{ \frac{1}{T} \sum^T_{t=1} \left( \bar\dv - \dv_t\right) \cdot \vec \delta_t(\x_{\star})}}{ \Sigma}\label{e:primal_dual_with_extra_term}
\end{align}
We provide two approaches to bound the r.h.s. term $\Sigma$ in Eq.~\eqref{eq:_a1}, and this gives the two conditions in Assumption~\ref{a:5}:

\noindent\textbf{Bound 1.} We can bound the r.h.s. term $\Sigma$ in the above equation  as follows
\begin{align}
\Sigma &= { \bar\dv\cdot  \sum^T_{t=1} \vec\delta_t(\x_{\star})}- { { \sum^T_{t=1} \dv_t\cdot \vec\delta_t(\x_{\star}) }}= {-{ \sum^T_{t=1} \dv_t\cdot \vec\delta_t(\x_{\star}) }} \\
    &\leq\frac{1}{u_{\star, \min}}{\sum_{i \in \I} \sum^T_{t=1} \delta_{t,i}(\x_{\star})  \mathds{1}_{\set{\delta_{t,i}(\x_{\star}) \geq 0}}} = \BigO{\VT}.\label{eq:_a1}
\end{align}
\noindent\textbf{Bound 2.} We alternatively bound  $ \Sigma$ as follows
\begin{align}
     \Sigma &=  {   \sum^K_{k=1} \sum_{t \in \mathcal{T}_k} \left(\bar\dv - \dv_t \right) \cdot \vec \delta_t(\x_{\star})}=  {\sum^K_{k=1} \sum_{t \in \mathcal{T}_k} \left(\bar\dv - \dv_{\min\parentheses{\T_k}}\right) \cdot \vec \delta_t(\x_{\star})+   \sum^K_{k=1} \sum_{t \in \mathcal{T}_k} \left( \dv_{\min\parentheses{\T_k}} - \dv_{t}\right) \cdot \vec \delta_t(\x_{\star})}\nonumber\\
     &\leq   \Delta_{\alpha}\sum^K_{k=1} \norm{\sum_{t \in \T_k}  \vec \delta_t(\x_{\star})}_1 +  { u_{\max} \sum^K_{k=1} \sum_{t \in \mathcal{T}_k} \norm{ \dv_{\min\parentheses{\T_k}} - \dv_{t}}_1 }, \label{e:w_rhs_proof_0}
\end{align}
where $\Delta_\alpha = \max \set{\norm{\dv - \dv'}_\infty : \dv, \dv' \in \Theta}$. We bound the term  $\sum^K_{k=1} \sum_{t \in \mathcal{T}_k} \norm{ \dv_{\min\parentheses{\T_k}} - \dv_{t}}_1$ in the above equation as 
\begin{align}
     \sum^K_{k=1} \sum_{t \in \mathcal{T}_k} \norm{ \dv_{\min\parentheses{\T_k}} - \dv_{t}}_1 &\leq  L_{\Theta} \sum^K_{k=1} \eta_{\Theta,\,{\min \parentheses{\T_k}}}\sum_{t \in \mathcal{T}_k} (t- \min\parentheses{\T_k})\leq L_{\Theta} \sum^K_{k=1} \eta_{\Theta,\,{\min \parentheses{\T_k}}} \card{\T_k}^2\\
     &= L_{\Theta}\frac{u_{\star, \min}^{1 + \frac{1}{\alpha}}}{\alpha}  \sum^K_{k=1} \frac{\card{\T_k}^2}{\min \parentheses{\T_k}} \label{e:w_rhs_proof},
\end{align}
and replacing this upper-bound in Eq.~\eqref{e:w_rhs_proof_0} gives 
\begin{align}
    \Sigma \leq    \Delta_{\alpha} \sum^K_{k=1} \norm{\sum_{t \in \T_k}  \vec \delta_t(\x_{\star})}_1 + u_{\max} L_{\Theta}\frac{\alpha}{u_{\star, \min}^{1 + \frac{1}{\alpha}}}  \sum^K_{k=1} \frac{\card{\T_k}^2}{\vGroup{\min \parentheses{\T_k}}{\sum_{k' < k } \card{\T_k}+1}}  = \BigO{\WT}\label{eq:_a2}.
\end{align}
We combine Eq.~\eqref{eq:_a1}, Eq.~\eqref{eq:_a2}, and Eq.~\eqref{e:primal_dual_with_extra_term} to obtain 
\begin{align}
     \regret &\leq  \sup_{ \set{\u_t}^T_{t=1} \in {{\U^T}}} \set{\frac{1}{T} \parentheses{{\mathfrak{R}_{T, \X} + \mathfrak{R}_{T, \Theta}}} }+ \BigO{\frac{\min\set{\VT, \WT}}{T} }\\
     &\leq  \sup_{ \set{\u_t}^T_{t=1} \in {{\U^T}}} \set{\frac{1}{T} \parentheses{{1.5 \diam{\X} \sqrt{\sum_{t \in \T} \norm{\vec g_{\X,t}}^2_2} +\frac{\alpha}{u_{\star, \min}^{1 + \frac{1}{\alpha}}}  \sum^T_{t=1}\frac{ \norm{\vec g_{\Theta,t}}^2_2}{t} } }}+  \BigO{\frac{\min\set{\VT, \WT}}{T} }.\label{e:adaptive_bound}
\end{align}
The following upper bounds hold
\begin{align*}
    \norm{\vec g_{\Theta,t}}_2 &= \norm{ \parentheses{\frac{1}{\parentheses{-\theta_{t,i}}^{1/\alpha}} - \left(\vec u_t (\x_t)\right)}_{i \in \I}}_2 \leq \sqrt{I} \max\set{\frac{1}{u_{\star, \min}^{1/\alpha}} - u_{\min}, u_{\max} - \frac{1}{u_{\star, \max}^{1/\alpha}}}  = L_{\Theta},\\
     \norm{\vec g_{\X,t}}_2 &= {\norm{ \dv_t \cdot \partial_{\x}\vec u_t (\x_t)}}_2 \leq \frac{1}{u_{\star, \min}^{\alpha}}   \norm{\partial_{\x}\vec u_t (\x_t)}_2 \leq  \frac{L_{\X}}{u_{\star, \min}^{\alpha}}.
\end{align*}
Thus, the regret bound in Eq.~\eqref{e:adaptive_bound} can be upper bounded as
\begin{align*}
    \regret &= \frac{1}{T}  \sup_{ \set{\u_t}^T_{t=1} \in {{\U^T}}} \set{{{1.5\diam{\X} \frac{L_{\X} \sqrt{T}}{u_{\star, \min}^{\alpha}} +\frac{\alpha}{u_{\star, \min}^{1 + \frac{1}{\alpha}}} \sum^T_{t=1}\frac{L^2_\Theta}{t}}}} + \frac{\min{\set{\VT, \WT}}}{T} \\
    &\leq  \frac{1}{T} \sup_{ \set{\u_t}^T_{t=1} \in {{\U^T}}} \set{{1.5\diam{\X} \frac{L_{\X} \sqrt{T}}{u_{\star, \min}^{\alpha}} +\frac{\alpha}{u_{\star, \min}^{1 + \frac{1}{\alpha}}}  L^2_{\Theta} (\log(T) + 1) }} + \frac{\min{\set{\VT, \WT}}}{T} \\
    &= \BigO{\frac{1}{\sqrt{T}} + \frac{\min \set{\VT, \WT}}{T}}.
    \label{e:final_eq}
\end{align*}
This concludes the proof.
\end{proof}

\section{Proof of Theorem~\ref{theorem:lowerbound} (Lower Bound)}
\label{proof:lowerbound}
\begin{proof}
Consider a scenario with a single player $\I=\{1\}$, $\X = \set{x \in \reals, \abs{x} \leq 1}$, and the utility selected by an adversary at time slot $t \in \T$ is given by
\begin{align}
    u_t(x) =  w_t x + 1, \quad\text{where}~  w_t \in \set{-1,+1}.
\end{align}
The weight $w_{t}$ is selected in $\{-1,+1\}$ uniformly at random for $t \in \T$. A policy $\A$ selects the sequence of decisions $\set{x_t}^T_{t=1}$ and has the following fairness regret
\begin{align*}
    &\mathbb{E} \brackets{\max_{x \in \X}f_{\alpha}\parentheses{ \frac{1}{T}\sum^T_{t=1} u_t(x)} -f_{\alpha}\parentheses{ \frac{1}{T}\sum^T_{t=1} u_t(x_t)}  } \geq  \mathbb{E} \brackets{\max_{x \in \X}f_{\alpha}\parentheses{ \frac{1}{T}\sum^T_{t=1} u_t(x)}}  -\vGroup{f_{\alpha}\parentheses{\mathbb{E}\brackets{ \frac{1}{T}\sum^T_{t=1} u_t(x_t)}} }{\text{$= 0$}}\\
    &= \mathbb{E} \brackets{f_{\alpha}\parentheses{\max_{x \in \X} \frac{1}{T}\sum^T_{t=1} u_t(x)}}  =  \mathbb{E} \brackets{f_{\alpha}\parentheses{\frac{1}{T}\abs{\sum^T_{t=1} w_{t,1}}+ 1} } \stackrel{
    \mathrm{(a)}}{\geq}{\mathbb E \brackets{\frac{1}{T}\abs{\sum^T_{t=1} w_{t,1}}}} \parentheses{\frac{2^{1-\alpha}-1}{1-\alpha}}\stackrel{\mathrm{(b)}}{\geq} \frac{\parentheses{\frac{2^{1-\alpha}-1}{1-\alpha}}}{\sqrt{2T}}\\
    &= \Omega\parentheses{\frac{1}{\sqrt{T}}}.
\end{align*}
Inequality (a) is obtained considering $f_{\alpha} (x+1)$ is concave in $x$, $f_{\alpha}(0 + 1)=0$, and $f_{\alpha}(x+1) \geq f_{\alpha}(2) x$ for $x \in [0, 1]$. Inequality~(b) is obtained through Khintchine inequality. A lower bound on the fairness regret~\eqref{e:b_regret} can be established:
\begin{align}
  \regret  \geq \mathbb{E} \brackets{\max_{x \in \X}f_{\alpha}\parentheses{ \frac{1}{T}\sum^T_{t=1} u_t(x)} -f_{\alpha}\parentheses{ \frac{1}{T}\sum^T_{t=1} u_t(x_t)}  } = \Omega\parentheses{\frac{1}{\sqrt{T}}}.
\end{align}
This concludes the proof.
\end{proof}

\section{Proof of Corollary~\ref{corollary:stochastic}}
\label{proof:stochastic}

\begin{proof}
\noindent \textbf{Expected regret.} When the utilities are i.i.d., we have the following
\begin{align}
    \mathbb E \brackets{  \u_t (\x)} = \vec u, \forall t \in \T,
\end{align}
for some fixed utility $\u \in \U$. In the proof Theorem~\ref{proof:t:maintheorem}, in particular, in Eq.~\eqref{e:primal_dual_with_extra_term} it holds
\begin{align}
    {F_{\alpha}\left( {\frac{1}{T} \sum^T_{t=1}\u_t(\x_{\star})}\right) - F_{\alpha}\left( {\frac{1}{T} \sum^T_{t=1}\u_t(\x_t)}\right)} &\leq  {\frac{\mathfrak{R}_{T, \X}}{T} + \frac{\mathfrak{R}_{T, \Theta}}{T} + { \frac{1}{T} \sum^T_{t=1} \left(\dv_t - \bar\dv\right) \cdot \u_t(\x_{\star})}}.
\end{align}
Taking the expectation of both sides gives 
\begin{align}
        \mathbb E \brackets{F_{\alpha}\left( {\frac{1}{T} \sum^T_{t=1}\u_t(\x_{\star})}\right) - F_{\alpha}\left( {\frac{1}{T} \sum^T_{t=1}\u_t(\x_t)}\right)}  \leq  \mathbb E \brackets{ \frac{\mathfrak{R}_{T, \X}}{T} + \frac{\mathfrak{R}_{T, \Theta }}{T}} + \mathbb E \brackets{{ \frac{1}{T} \sum^T_{t=1} \left(\dv_t - \bar\dv\right) \cdot \u_t(\x_{\star})}}.
\end{align}
The variables $\dv_t$ and  $\u_t$ are independent for  $t \in \T$, thus we have
\begin{align}
     \mathbb E \brackets{{ \frac{1}{T} \sum^T_{t=1} \left(\dv_t - \bar\dv\right) \cdot \u_t(\x_{\star})}} = \mathbb E \brackets{{  \parentheses{\bar\dv - \bar\dv }\cdot \u(\x_{\star})}} = 0.
\end{align}
Through Eq.~\eqref{e:final_eq}, it holds
\begin{align}
    \mathbb E \brackets{F_{\alpha}\left( {\frac{1}{T} \sum^T_{t=1}\u_t(\x_{\star})}\right) - F_{\alpha}\left( {\frac{1}{T} \sum^T_{t=1}\u_t(\x_t)}\right)} = \BigO{\frac{1}{\sqrt{T}}}.
\end{align}
This concludes the first part of the proof.

\noindent\textbf{Almost-sure zero-regret.} Let $\Delta =  \parentheses{u_{\max} - u_{\min}}$,  $\T = \T_1\cup\T_2\cup \dots \cup \T_K$ where $K = T^{2/3}$ and $\card{\T_k} = \kappa  = T^{1/3}$ for $k \in \set{1,2,\dots, K}$, and let $\beta \in (0,1/6)$. Employing Hoeffding's inequality we can bound the l.h.s. term in Eq.~\eqref{e:adv1} for $i \in \I$ as 
\begin{align}
\mathbb P \parentheses{  \abs{\sum_{t \in \T_k} \delta_{t,i} (\x)} \leq \Delta T^{1/6+\beta}} &\geq 1 -2 \exp\parentheses{\frac{-2T^{1/3+2\beta}}{\parentheses{(T-\kappa) \kappa^2/T^2 + \kappa (\kappa / T -1)^2}}} = 1 -2 \exp\parentheses{\frac{-2T^{1/3+2\beta}}{\parentheses{\kappa -\kappa^2 / T}}} \\
&=  1 - 2\exp\parentheses{\frac{-2T^{1/3+2\beta}}{\parentheses{T^{1/3} -T^{ - 1/3}}}}.
\end{align}
Hence, it follows 
\begin{align*}
   \mathbb{P} \parentheses{\sum^K_{k=1} \sum_{i \in \I} \abs{\sum_{t \in \T_k}   \delta_{t,i} (\x)} \leq \Delta T^{5/6+\beta}} &\geq \parentheses{1 -2 \exp\parentheses{\frac{-2T^{1/3+2\beta}}{\parentheses{T^{1/3} -2 T^{ - 1/3}}}}}^{I T^{2/3}}\\
   &\geq1 - 2I  T^{2/3} \exp\parentheses{\frac{-2T^{1/3+2\beta}}{\parentheses{T^{1/3} -T^{ - 1/3}}}} &\text{Bernoulli's inequality}\\
   & \geq 1 -2 I  T^{2/3} \exp\parentheses{\frac{-2T^{1/3+2\beta}}{{T^{1/3}}}}\\
   &\geq  1 - 2I  T^{2/3} \exp(-2T^{2\beta}).
\end{align*}
It follows from  the above equation paired with Eq.~\eqref{e:adv1}
\begin{align}
    \WT = \BigO{T^{5/6+\beta} + T^{2/3}} = \BigO{T^{5/6+\beta}}, \qquad \text{w.p.}\qquad p \geq  1 - 2I  T^{2/3} \exp(-2T^{2\beta}).
\end{align}
Thus, for any $\beta \in (0, 1/6)$ and $T\to \infty$, it holds
\begin{align}
    \frac{\WT}{T} \leq  0, \qquad \text{w.p.}\qquad  p \geq 1.
\end{align}
Note that given that $\WT \geq 0$ in Eq.~\eqref{e:adv2}, it holds $ \lim_{T \to \infty}  \WT = 0$ w.p. $p = 1$.   Therefore, it follows from Theorem~\ref{proof:t:maintheorem} for $T \to \infty$
\begin{align}
       \regret = \BigO{\frac{1}{\sqrt{T}} + \frac{\min\set{\VT, \WT}}{T}} = \BigO{\frac{1}{\sqrt{T}} + \frac{\WT}{T}} \leq 0, \qquad \text{w.p.}\qquad 1.
\end{align} 
This concludes the proof.
\end{proof}

\section{Additional Experimental Details}
\begin{table}[h]
	\caption{Specification of the network topologies used in experiments.}
	\label{t:setting}
	\begin{footnotesize}
	\begin{center}
	
		\begin{tabular}{|c|c|c|c|c|c|c|c|}
			\hline
			Topologies & $\card{\C}$ & $\card{\E}$ & $k_c$ & $\card{\mathcal{Q}_i}$ & $\card{\cup_{f \in \F} \Lambda_f(\C)}$ & $w$  & Figure              \\
			\hline
			\SC        & 3           & 3           & 5--5  & 1                      & 1                                      & 1--2 & Fig.~\ref{fig:topologies}~(a) \\ 
			\BT-1--\BT-3        & 13          & 12          & 1--5  & 2--5                      & 1                                      & 1--9 & Fig.~\ref{fig:topologies}~(b)--(d) \\
			\Grid      & 9           & 12          & 1--5  & 2                      & 1                                      & 1--7 & Fig.~\ref{fig:topologies}~(e) \\
			\Abilene   & 12          & 13          & 1--5  & 2                      & 2                                      & 1--8 & Fig.~\ref{fig:topologies}~(f) \\
			\GEANT     & 22          & 33          & 1--5  & 3                      & 2                                      & 1--9 & Fig.~\ref{fig:topologies}~(g) \\
			\hline
		\end{tabular}
	\end{center}
	\end{footnotesize}
	\label{table:topologies}
\end{table} 
\begin{figure}[h]
    \centering
    \subcaptionbox{Stationary }{\includegraphics[width=.2\linewidth]{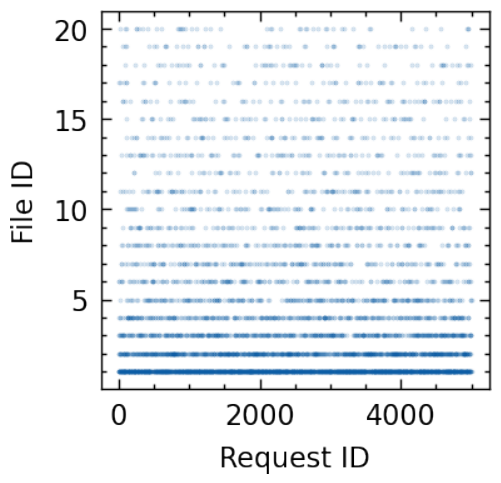}}
       \subcaptionbox{Non-Stationary }{\includegraphics[width=.2\linewidth]{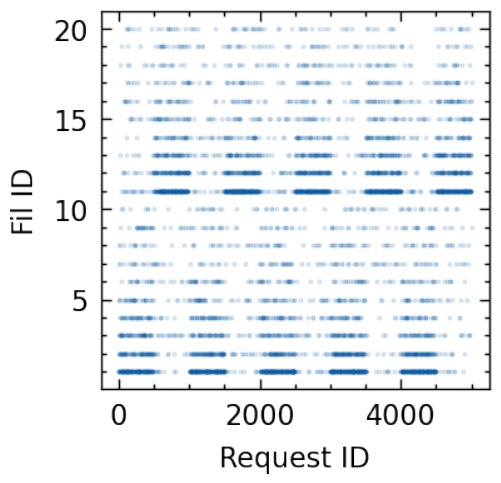}}
    \caption{Request traces stationary~(a) and non-stationary~(b) configured with $\sigma = 1.2$, $T = 5000$, $F = 20$, $D = 100$. Each dot indicates a requested file.}
    \label{fig:traces}
\end{figure}
\section{Departing and Arriving Agents}
\label{appdendix:departingarriving}
\new{ The system model in Section~\ref{s:system_model} supports departing and arriving agents. Consider a population of agents $\mathcal I$, the system may only observe  a subset of the agents as the \emph{participating} agents, $\mathcal I_t \subset \mathcal I$ at time $t$, and the utility of \emph{absent} agents is simply $u_{t,i}(\,\cdot\,) = 0$. For example, in the extreme scenario where a single agent $t \in \mathcal T$ is participating at a given time slot, the long-term fairness objective~\eqref{eq:hf_objective} falls back to the slot-fairness objective~\eqref{eq:sf_objective}, i.e.,  $F_\alpha\left({\sum_{t \in \mathcal T}} \vec u_t(\vec x_t)\right) = \sum_{t\in \mathcal T} f_\alpha \left( u_{t,t} (\vec x_t)\right)$ where the fairness is ensured across the different agents arriving at different timeslots $t \in \mathcal T$. It is easy to verify that even in the case when the set of agents $\mathcal I$ is unknown to the controller in advance, one could augment the dual space with an extra dimension each time a new user appears, and the same guarantees hold.}

\section{Time-Complexity of Algorithm~\ref{alg:primal_dual_ogaogd}}
\label{appdendix:time_complexity}
\new{Algorithm~1 applied to the virtualized caching system application has a time complexity $\BigO{C F^2}$, where $C$ is number  of caches and $F$ is the number of files in the catalog; the most expensive operation in Algorithm~1 is the projection step in line 8 that corresponds to the Euclidean projection onto a capped simplex, and this can be performed in $\BigO{F^2}$ steps~\cite{wang2015projection} for each cache state.  Despite the  high time complexity ($F$ is typically large),  in practice solvers (e.g., CVXPY~\cite{diamond2016cvxpy}) support \emph{warm-start} that speeds up the projection when the warm-start parameters are close to the ones of obtained by the solution, and since Algorithm~1 is iterative and the cache states do not severely change, typically a lower computational cost is achieved. Moreover, the proposed caching model in Section~\ref{s:experiments} supports request batching, where a batch
includes the requests arriving between two consecutive cache updates. Batching amortizes the computational
cost of the different policies, reducing the cost per request by the batch size ($R_t$ at time slot $t$).}
\end{document}